\documentclass[screen]{llncs}

\makeatletter
\renewcommand\paragraph{\@startsection{paragraph}{4}{\z@}%
  {2.25ex \@plus 1ex \@minus .2ex}%
  {-0.75em}%
  {\normalfont\normalsize\bfseries}}
\makeatother

\usepackage{etoolbox}
\newtoggle{arxiv}
\toggletrue{arxiv}

\usepackage[utf8]{inputenc}
\usepackage[utf8]{inputenc}
\usepackage{color}
\usepackage[usenames,dvipsnames]{xcolor}

\usepackage{adjustbox}

\usepackage{eucal}
\usepackage{centernot}

\usepackage{thm-restate}

\usepackage{siunitx}
\usepackage{textcomp}
\usepackage{caption}

\usepackage{pst-func}
\usepackage{amssymb}
\usepackage{graphicx}

\usepackage{paralist}
\usepackage{environ}
\usepackage{pgfplots}
\usepackage{tikz}
\usetikzlibrary{shadows,arrows,shapes,snakes,automata,backgrounds,petri}

\usepackage{centernot} 

\usepackage{mathtools}
\usepackage{array,multirow}

\usepackage{url}

\usepackage{pgfplotstable}
\usepackage{mdframed}
\usepackage{threeparttable}

\usepackage{etoolbox}
\usepackage{arydshln}

\usepackage{xspace} 

\usepackage{dirtytalk} 
\usepackage{stmaryrd} 

\usepackage[inline,shortlabels]{enumitem}
\newlist{inlinelist}{enumerate*}{1}
\setlist*[inlinelist,1]{%
  label=(\roman*),
}

\usepackage{xifthen}  
\usepackage{ifthen}

\usepackage{nicefrac}

\usepackage{hyperref}
\usepackage{cleveref}
\usepackage{bigints}

\usepackage{caption,subcaption}     

\usetikzlibrary{pgfplots.dateplot}
\usetikzlibrary{matrix,chains,positioning,decorations.pathreplacing}
\usetikzlibrary{patterns,calc,fit,arrows}
\usetikzlibrary{shapes.geometric}

\usepackage[final,nomargin,inline,index]{fixme} 
\fxusetheme{color}

\FXRegisterAuthor{bart}{anbart}{\color{magenta} {\underline{bart}}}
\FXRegisterAuthor{ste}{anste}{\color{blue} {\underline{ste}}}
\FXRegisterAuthor{lillo}{anlillo}{\color{orange} {\underline{lillo}}}
\FXRegisterAuthor{MM}{anMM}{\color{green} {\underline{Maurizio}}}

\hypersetup{
  breaklinks   = true,
  colorlinks   = true, 
  urlcolor     = blue, 
  linkcolor    = blue, 
  citecolor    = red   
}
\hypersetup{final} 

\usepackage{listings}

\definecolor{listingBG}{HTML}{FFFFCB}%
\definecolor{listingFrame}{HTML}{BBBB98}%
\definecolor{listingLineno}{rgb}{0.5,0.5,1.0}%

\definecolor{LightGrey}{rgb}{0.975,0.975,0.975}

\lstdefinelanguage{tins}{
	commentstyle=\color{Gray},
	morecomment=[l]{//},
	morecomment=[s]{/*}{*/},
	classoffset=0,
	morekeywords={skip,throw,if,then,else,while,do},
	keywordstyle=\color{Black}\bfseries,
	classoffset=1,
	morekeywords={sender,value,balance,undef},
	keywordstyle=\valColor{}\bfseries\itshape,
	classoffset=2,
	morekeywords={contract},
	keywordstyle=\funColor{}\bfseries,
        extendedchars=true,
        literate={\$}{{\textbf{{\dollar}}}}1
}

\definecolor{LightGrey}{rgb}{0.975,0.975,0.975}

\lstdefinelanguage{solidity}{
	commentstyle=\color{Gray},
	morecomment=[l]{//},
	morecomment=[s]{/*}{*/},
	classoffset=0,
        escapechar=\$,
	morekeywords={struct,mapping,function,this,public,private,static,final,class,extends,switch,case,break,finally,try,catch,return,if,else,new},
	keywordstyle=\color{Blue}\bfseries,
	classoffset=1,
	morekeywords={unit,int,string,bool,address,uint,uint256},
	keywordstyle=\color{TealBlue},
	classoffset=2,
	morekeywords={ether,wei,finney,contract,send,throw,msg,sender,value},
	keywordstyle=\color{Plum}\bfseries,
}

\usepackage{tikz-cd}
\usetikzlibrary{cd,decorations.markings}

\newcommand{\ifempty}[3]{%
  \ifthenelse{\isempty{#1}}{#2}{#3}%
}

\newcommand{\ifzero}[3]{%
  \ifthenelse{\equal{#1}{0}}{#2}{#3}%
}

\newcommand{\ifdots}[3]{%
  \ifthenelse{\equal{#1}{...}}{#2}{#3}%
}

\newcommand{\hidden}[1]{}

\newcommand{\eqdef}{\triangleq}

\renewcommand{\vec}[1]{\boldsymbol{#1}}
\newcommand{\relR}{\mathcal{R}}

\newcommand{\Real}[1]{\mathrm{Real}}

\newcommand{\codefont}{\fontsize{9}{9}\selectfont}
\newcommand{\code}[1]{{\tt\codefont {#1}}}
\newcommand{\codeAddr}[1]{\ensuremath{\code{\aColor{#1}}}}
\newcommand{\codeVal}[1]{\ensuremath{\code{\valColor{#1}}}}
\newcommand{\codeFun}[1]{\ensuremath{\code{\funColor{#1}}}}

\newcommand{\dollar}{{\textup{\texttt{\symbol{`\$}}}}}
\newcommand{\codeand}{\textup{\texttt{\symbol{`\&}\symbol{`\&}}}}

\newcommand{\eg}{e.g.\@\xspace}
\newcommand{\ie}{i.e.\@\xspace}
\newcommand{\wrt}{w.r.t.\@\xspace}

  {}

\newcommand{\BTC}{\textup{%
  \leavevmode
  \vtop{\offinterlineskip 
    \setbox0=\hbox{B}%
    \setbox2=\hbox to\wd0{\hfil\hskip-.03em
    \vrule height .3ex width .15ex\hskip .08em
    \vrule height .3ex width .15ex\hfil}
    \vbox{\copy2\box0}\box2}}\xspace}

\newcommand{\true}{\code{true}\xspace}

\def\aColor{\color{ForestGreen}}
\def\pColor{\color{ForestGreen}}
\def\cColor{\color{ForestGreen}}

\newcommand{\Addr}{{\aColor{\textup{\textbf{Addr}}}}}
\newcommand{\Val}{{\valColor{\textup{\textbf{Val}}}}}
\newcommand{\Const}{{\valColor{\textup{\textbf{Const}}}}}
\newcommand{\Tx}{{\valColor{\textup{\textbf{Tx}}}}}

\newcommand{\addrToContr}[1]{{\Gamma}\ifempty{#1}{}{({#1})}}

\newcommand{\aFmt}[1]{{\aColor{\mathcal{#1}}}}
\newcommand{\pFmt}[1]{{\pColor{\mathcal{#1}}}}
\newcommand{\cFmt}[1]{{\cColor{\mathcal{#1}}}}

\newcommand{\amv}[2][]{\aFmt{#2}_{\aColor{#1}}\xspace}
\newcommand{\amvA}[1][]{\amv[{#1}]{X}}
\newcommand{\amvB}[1][]{\amv[{#1}]{Y}}

\newcommand{\pmv}[2][]{\pFmt{#2}_{\pColor{#1}}\xspace}
\newcommand{\pmvA}[1][]{\pmv[{#1}]{A}}
\newcommand{\pmvB}[1][]{\pmv[{#1}]{B}}

\newcommand{\cmv}[2][]{\cFmt{#2}_{\cColor{#1}}\xspace}
\newcommand{\cmvA}[1][]{\cmv[{#1}]{C}}
\newcommand{\cmvB}[1][]{\cmv[{#1}]{D}}

\newcommand{\qmv}[2][]{{#2}_{#1}}
\newcommand{\qmvA}[1][]{\qmv[{#1}]{p}}

\newcommand{\qmvB}[1][]{\qmv[{#1}]{q}}

\newcommand{\QmvA}[1][]{P_{#1}}
\newcommand{\QmvAi}[1][]{P'_{#1}}
\newcommand{\QmvB}[1][]{Q_{#1}}

\newcommand{\QmvC}[1][]{R_{#1}}

\newcommand{\QmvU}[1][]{\mathbb{P}}

\definecolor{BlueViolet}{rgb}{0, 0, 0.55}
\definecolor{RubineRed}{rgb}{0.88, 0.07, 0.37}
\definecolor{ForestGreen}{rgb}{0.13, 0.55, 0.13}
\definecolor{NavyBlue}{rgb}{0.0, 0.0, 0.5}
\definecolor{Black}{rgb}{0.02, 0.02, 0.02}
\definecolor{MidnightBlue}{rgb}{0.0, 0.2, 0.4}
\definecolor{Gray}{rgb}{0.41, 0.41, 0.41}

\def\txColor{\color{MidnightBlue}}
\newcommand{\txFmt}[1]{{\txColor{\sf #1}}}

\newcommand{\tx}[2][]{\txFmt{#2}_{\txColor{#1}}}
\newcommand{\txT}[1][]{\tx[#1]{T}}
\newcommand{\txTi}[1][]{\txFmt{T'_{\txColor{{#1}}}}}

\DeclareMathAlphabet{\mathbfsf}{\encodingdefault}{\sfdefault}{bx}{n}

\newcommand{\bcEmpty}{{\mathbfsf{\txColor{\epsilon}}}}
\newcommand{\bcB}[1][]{{\boldsymbol{\mathsf{\txColor{B}}}}_{\txColor{#1}}}
\newcommand{\bcBi}[1][]{{\boldsymbol{\txColor{\sf B'_{\txColor{\mathrm{\textup{#1}}}}}}}}

\newcommand{\bcSt}[1][]{\sigma_{#1}}
\newcommand{\bcSti}[1][]{\sigma_{#1}'}
\newcommand{\bcStii}[1][]{\sigma_{#1}''}
\newcommand{\bcStInit}{\sigma^{\star}}

\newcommand{\bcEnv}[1][]{\rho_{#1}}

\newcommand{\bcKV}[1][]{\delta_{#1}}

\newcommand{\balanceInit}[1]{\valN^{\star}_{#1}}

\newcommand{\mrg}{\oplus}
\newcommand{\mSubst}[1][]{\pi_{#1}}
\newcommand{\mSubsti}[1][]{\pi'_{#1}}
\newcommand{\WR}[1]{\ifempty{#1}{\Pi}{\Pi(#1)}}
\newcommand{\WRi}[1]{\ifempty{#1}{\Pi'}{\Pi'(#1)}}
\newcommand{\WRmin}[1]{\ifempty{#1}{\Pi^{\star}}{\Pi^{\star}(#1)}}

\newcommand{\trans}[1]{\xrightarrow{#1}}
\newcommand{\nottrans}[1]{\centernot{\xrightarrow{#1}}}

\newcommand{\pre}[1]{{}^{\bullet}{#1}}

\newcommand{\post}[1]{{#1}{{}^{\bullet}}}

\newcommand{\irule}[2]{\dfrac{#1}{#2}}

\newcommand{\mapstopart}{\rightharpoonup}

\newcommand{\nrule}[1]{{\scriptsize \textsc{#1}}}
\newcommand{\smallnrule}[1]{{\tiny \textsc{#1}}}

\newcommand{\sem}[2][]{\mbox{\ensuremath{\llbracket{#2}\rrbracket_{#1}}}}

\newcommand{\msem}[3]{\mbox{\ensuremath{\llbracket{#3}\rrbracket^{#1}_{#2}}}}

\newcommand{\semCmd}[3]{\mbox{\ensuremath{\llbracket{#1}\rrbracket_{#2}^{#3}}}}
\newcommand{\semTx}[2]{\sem[#2]{{#1}}}
\newcommand{\semBc}[2]{\sem[#2]{{#1}}}

\newcommand{\equivStSeq}[1][]{\simeq_{#1}}
\newcommand{\equivSeq}{\simeq}

\newcommand{\dom}[1]{\operatorname{dom} {#1}}
\newcommand{\keys}[1]{\operatorname{keys}({#1})}

\newcommand{\Nat}{\mathbb{N}}

\newcommand{\bind}[2]{\nicefrac{#2}{#1}}
\newcommand{\setenum}[1]{\{#1\}}
\newcommand{\setcomp}[2]{\left\{{#1} \,\middle|\, {#2}\right\}}
\newcommand{\emptymset}{[]}
\newcommand{\msetenum}[1]{\lbrack{#1}\rbrack}

\newcommand{\card}[1]{|#1|}

\newcommand{\qedef}{\ensuremath{\diamond}}
\newcommand{\qedhere}{\ensuremath{\diamond}}

\crefname{appendix}{appendix}{appendices}
\Crefname{appendix}{Appendix}{Appendices}
\crefname{notation}{notation}{notations}
\Crefname{notation}{Notation}{Notations}

\definecolor{LightGrey}{rgb}{0.95,0.95,0.95}
\definecolor{keyword}{HTML}{7F0055}

\newlength\replength
\newcommand\repfrac{.1}

\newcommand\rulewidth{.6pt}
\newcommand\tdashfill[1][\repfrac]{\cleaders\hbox to \replength{%
  \smash{\rule[\arraystretch\ht\strutbox]{\repfrac\replength}{\rulewidth}}}\hfill}

\newcommand\tdotfill[1][\repfrac]{\cleaders\hbox to \replength{%
  \smash{\raisebox{\arraystretch\dimexpr\ht\strutbox-.1ex\relax}{.}}}\hfill}

\newcommand{\var}[2][]{#2_{#1}} 
\newcommand{\varX}[1][]{\var[#1]{x}} 
 
\newcommand{\varY}[1][]{\var[#1]{y}}

\def\funColor{\color{RubineRed}}
\newcommand{\funFmt}[1]{{\funColor{\mathtt{#1}}}}

\newcommand{\funF}[1][]{\funFmt{f}_{\funColor{#1}}} 
 
\newcommand{\funG}[1][]{\funFmt{g}_{\funColor{#1}}} 
 
\newcommand{\funH}[1][]{\funFmt{h}_{\funColor{#1}}}

\def\valColor{\color{BlueViolet}}
\newcommand{\val}[2][]{{\valColor{#2_{#1}}}} 
\newcommand{\valV}[1][]{\val[#1]{v}}
\newcommand{\valVi}[1][]{\val[#1]{v'}}
\newcommand{\valN}[1][]{\val[#1]{n}}
\newcommand{\valNi}[1][]{\val[#1]{n'}}
\newcommand{\valK}[1][]{\val[#1]{k}}
\newcommand{\valKi}[1][]{\val[#1]{k'}}
\newcommand{\valX}[1][]{\val[#1]{x}}
\newcommand{\valY}[1][]{\val[#1]{y}}

\newcommand{\valZ}[1][]{\val[#1]{z}}

\def\cmdColor{\color{RubineRed}}
\newcommand{\cmdFmt}[1]{{\cmdColor{\mathit{#1}}}}
\newcommand{\cmdC}[1][]{\mathord{\cmdFmt{S}_{\cmdColor{#1}}}}

\newcommand{\cmdSkip}{\ensuremath{\code{skip}}}
\newcommand{\cmdAss}[2]{{#1} \code{:=} {#2}}
\newcommand{\cmdIfTE}[3]{\code{if}\, {#1}\, \code{then} \, {#2} \, \code{else} \, {#3}}
\newcommand{\cmdIfT}[2]{\code{if} \, {#1} \, \code{then} \, {#2}}

\newcommand{\cmdCall}[4][]{{#2}\ifempty{#3}{}{:{#3}({#4})}\ifempty{#1}{}{\dollar {#1}}}
\newcommand{\cmdSend}[2]{{#2}.\texttt{transfer}({#1}){}{}}
\newcommand{\cmdThrow}{\code{throw}\xspace}

\newcommand{\expFmt}[1]{{\cmdColor{\mathit{#1}}}}

\newcommand{\expEi}[1][]{\mathord{\expFmt{E'}}}
\newcommand{\expEii}[1][]{\mathord{\expFmt{E''}}}

\newcommand{\expGet}[2]{{#1}.{#2}} 

\newcommand{\expLookup}[1]{{#1}}

\definecolor{LightGrey}{rgb}{0.95,0.95,0.95}
\definecolor{keyword}{HTML}{7F0055}

\newmdenv[linewidth=0pt]{mdNoFramed}
\makeatletter

\newcommand*{\tabminted@finalstrut}[1]{%
  \ifdim\prevdepth>0pt
    \ifdim\dp#1>\prevdepth
      \vskip\dimexpr(\dp#1)-\prevdepth\relax
    \fi
  \else
    \vskip\dimexpr(\dp#1)\relax
  \fi
}
\newcommand*{\@tabmintedend}{%
  \let\@finalstrut\tabminted@finalstrut
}
\makeatother

\newcommand{\netN}{\sf{N}}
\newcommand{\Places}{{\sf P}}

\newcommand{\Transitions}{{\sf Tr}}
\newcommand{\Arcs}{{\sf{F}}}
\newcommand{\markM}[1][]{{\sf{m}_{#1}}}
\newcommand{\markMi}[1][]{{\sf{m}'_{#1}}}
\newcommand{\markMii}[1][]{{\sf{m}''_{#1}}}
\newcommand{\markMiii}[1][]{{\sf{m}}'''_{#1}}
\newcommand{\trT}[1][]{\mathsf{t}_{#1}}
\newcommand{\trTi}[1][]{{{\sf{t}}_{#1}'}}
\newcommand{\trTii}[1][]{{{\sf{t}}_{#1}''}}
\newcommand{\plP}[1][]{{\sf{p_{#1}}}}

\newcommand{\placeP}[1][]{{\sf{p_{#1}}}}
\newcommand{\placePi}[1][]{{\sf{p'_{#1}}}}
\newcommand{\PNet}[2]{\ifempty{#2}{{\sf N}_{{#1}}}{{\sf N}_{#1}(#2)}}
\newcommand{\trSU}[1][]{{\sf U}_{#1}}
\newcommand{\trSUi}[1][]{{\sf U'}_{#1}}
\newcommand{\trSUii}[1][]{{\sf U''}_{#1}}

\newcommand{\varBalance}{\ensuremath{\codeVal{balance}}\xspace}
\newcommand{\varSender}{\ensuremath{\codeVal{sender}}\xspace}

\newcommand{\varValue}{\ensuremath{\codeVal{value}}\xspace}

\newcommand{\ethtx}[5]{{#2}\xrightarrow{#1} {#3}:{#4}({#5})}

\newcommand{\contrFun}[3]{{#1}({#2}) \{ {#3} \}}
\newcommand{\contrFunSig}[2]{{#1}({#2})}

\newcommand{\swap}{\rightleftarrows}
\newcommand{\nswap}{\mbox{\ensuremath{\,\not\rightleftarrows\,}}}

\newcommand{\safeapprox}[3][]{{#2} \models^{#1} {#3}}
\newcommand{\wapprox}[2]{\safeapprox[w]{#1}{#2}}
\newcommand{\rapprox}[2]{\safeapprox[r]{#1}{#2}}

\newcommand{\pswap}[2]{\#^{#1}_{#2}}
\newcommand{\pswapWR}{\pswap{W}{R}}

\newcommand{\rset}[1]{R\ifempty{#1}{}{({#1})}}

\newcommand{\wset}[1]{W\ifempty{#1}{}{({#1})}}

\newcommand{\txsA}[1][]{{\txColor{\mathbb{T}}_{\txColor{#1}}}}
\newcommand{\txsAi}[1][]{{\txColor{\mathbb{T}'}_{\txColor{#1}}}}

\newcommand{\bcsA}[1][]{{\txColor{\mathbb{B}}_{\txColor{#1}}}}

\newcommand{\seqn}{\vartriangleleft}

\newcommand{\independent}{\;\mathrm{I}\;}

\newcommand{\txOfTr}[1]{\alpha\ifempty{#1}{}{({#1})}}
\newcommand{\trOfTrSU}[1]{{\it tr}\ifempty{#1}{}{({#1})}}

\newcommand{\mytitle}{A true concurrent model of \\ smart contracts executions}

\begin{document}

\title{\mytitle}

\author{Massimo Bartoletti\inst{1} \and
Letterio Galletta\inst{2} \and
Maurizio Murgia\inst{3}
}
\institute{University of Cagliari, Italy \and
IMT School for Advanced Studies, Lucca, Italy \and
University of Trento, Italy}

\maketitle

\begin{abstract}
  The development of blockchain technologies has enabled the trustless execution 
  of so-called \emph{smart contracts}, \ie programs that regulate 
  the exchange of assets (\eg, cryptocurrency) between users.
  In a decentralized blockchain,
  the state of smart contracts is collaboratively maintained
  by a peer-to-peer network of mutually untrusted nodes,
  which collect from users a set of \emph{transactions}
  (representing the required actions on contracts),
  and execute them in some order.
  Once this sequence of transactions is appended to the blockchain, 
  the other nodes validate it, 
  re-executing the transactions in the same order.
  The serial execution of transactions does not take advantage of 
  the multi-core architecture of modern processors, 
  so contributing to limit the throughput.
  In this paper 
  we propose a true concurrent model of smart contracts execution. 
  Based on this, we show how static analysis of smart contracts 
  can be exploited to parallelize the execution of transactions.
\end{abstract}

\section{Introduction}
\label{sec:intro}

Smart contracts~\cite{Szabo97firstmonday} are computer programs 
that transfer digital assets between users without a trusted authority.
Currently, smart contracts are supported by several blockchains,
the first and most widespread one being Ethereum~\cite{ethereum}.
Users interact with a smart contract by sending \emph{transactions},
which trigger state updates, and may possibly
involve transfers of crypto-assets between the called contract and the users.
The sequence of transactions on the blockchain determines the
state of each contract, and the balance of each user.

The blockchain is maintained by a peer-to-peer network of nodes,
which follow a consensus protocol to determine, at each turn, 
a new block of transactions to be added to the blockchain.
This protocol guarantees the correct execution of contracts
also in the presence of (a minority of) adversaries in the network,
and ensures that all the nodes have the same view of their state.
Nodes play the role of \emph{miner} or that of \emph{validator}.
Miners gather from the network sets of transactions sent by users, 
and execute them \emph{serially} to determine the new state.
Once a block is appended to the blockchain, 
validators re-execute all its transactions, 
to update their local view of the contracts state and of the users' balance.
To do this, validators process the transactions exactly in the same order 
in which they occur in the block, since choosing a different order could potentially
result in inconsistencies between the nodes
(note that miners also act as validators, since they validate all the blocks received from the network).

Although executing transactions in a purely sequential fashion
is quite effective to ensure the consistency of the blockchain state, 
in the age of multi-core processors 
it fails to properly exploit the computational capabilities of nodes.
By enabling miners and validators to concurrently execute transactions, 
it would be  possible to improve the efficiency and the throughput of the blockchain.

This paper exploits techniques from concurrency theory to 
provide a formal backbone for parallel executions of transactions.
More specifically, our main contributions can be summarised as follows:
\begin{itemize}

\item As a first step, we formalise blockchains,
   giving their semantics as a function which maps each contract to its state,
   and each user to her balance.
   This semantics reflects the standard implementation of nodes,
   where transactions are evaluated in sequence, without any concurrency.

\item We introduce two notions of \emph{swappability} of transactions.
  The first is purely semantic: two adjacent transactions can be swapped
  if doing so preserves the semantics of the blockchain.
  The second notion, called \emph{strong} swappability, is more syntactical:
  it checks a simple condition 
  (inspired by Bernstein's conditions~\cite{bernstein})
  on static approximations of the variables read/written by the transactions.
  \Cref{th:pswap-implies-swap} shows that strong swappability
  is strictly included in the semantic relation.
  Further, if we transform a blockchain
  by repeatedly exchanging adjacent strongly swappable transactions,
  the resulting blockchain is observationally equivalent to the original one
  (\Cref{th:pswapWR:mazurkiewicz}).

\item Building upon strong swappability, we devise
  a true concurrent model of transactions execution.
  To this purpose, we transform a block of transactions into an
  \emph{occurrence net}, 
  describing exactly the partial order induced by the swappability relation.
  We model the concurrent executions of a blockchain
  in terms of the \emph{step firing sequences} 
  (\ie finite sequences of \emph{sets} of transitions)
  of the associated occurrence net.
  \Cref{th:bc-to-pnet} establishes that 
  the concurrent executions and the serial one
  are semantically equivalent.

\item We describe how miners and validators can use our results to
  concurrently execute transactions, exploiting the multi-core architecture
  available on their nodes. 
  Remarkably, our technique is compatible with 
  the current implementation of the Ethereum blockchain,
  while the other existing approaches to parallelize transactions execution
  would require a soft-fork.

\item We apply our technique to ERC-721 tokens, 
  one of the most common kinds of contracts in Ethereum, 
  showing them to be suitable for parallelization.

\end{itemize}

\noindent
\iftoggle{arxiv}
{The proofs of our results are in~\Cref{sec:proofs}.}
{Because of space constraints, all the proofs of our results are in~\cite{BGM20arxiv}.}

\section{Transactions and blockchains}
\label{sec:transactions}

In this~\namecref{sec:transactions} we introduce a general model
of transactions and blockchains, abstracting from the actual
smart contracts language.

A smart contract is a finite set of functions, \ie terms of the form
$\contrFun{\funF}{\vec{\varX}}{\cmdC}$,
where $\funF$ is a function name,
$\vec{\varX}$ is the sequence of formal parameters (omitted when empty),
and $\cmdC$ is the function body.
We postulate that the functions in a contract have distinct names.
We abstract from the actual syntax of $\cmdC$,
and we just assume that the semantics of function bodies is defined
(see \eg \cite{BGM19cbt} for a concrete instance of syntax and semantics
of function bodies).

Let $\Val$ be a set of \emph{values},
ranged over by $\valV, \valVi, \ldots$,
let $\Const$ be a set  of \emph{constant names} $\varX, \varY, \ldots$,
and let $\Addr$ be a set of \emph{addresses} $\amvA,\amvB,\ldots$,
partitioned into \emph{account addresses} $\pmvA, \pmvB, \ldots$
and \emph{contract addresses} $\cmvA, \cmvB, \ldots$.
We assume a mapping $\addrToContr{}$ from addresses to contracts.

We assume that each contract has a key-value store, 
which we render as a partial function $\Val \mapstopart \Val$
from keys $\valK \in \Val$ to values.
The state of the blockchain is a function
$\bcSt : \Addr \rightarrow (\Val \mapstopart \Val)$
from addresses to key-value stores.
We postulate that 
$\varBalance \in \dom{\bcSt \amvA}$ for all $\amvA$.
A \emph{qualified key} is a term of the form $\amvA.\valK$.
We write $\bcSt (\amvA.\valK)$ for $\bcSt \amvA \valK$;
when $\valK \not\in \dom{\bcSt \, \amvA}$,
we write $\bcSt \, (\amvA . \valK) = \bot$.
We use $\qmvA,\qmvB,\hdots$ to range over qualified keys,
$\QmvA,\QmvB,\hdots$ to range over sets of them,
and $\QmvU$ to denote the set of all qualified keys.

To have a uniform treatment of accounts and contracts,
we assume that for all account addresses $\pmvA$,
$\dom{\bcSt \pmvA} = \setenum{\varBalance}$,
and that the contract $\addrToContr{\pmvA}$ has exactly one function, 
which just skips.
In this way, the statement $\cmdSend{n}{\pmvA}$,
which transfers $n$ currency units to $\pmvA$,
can be rendered as a call to this function.

\emph{State updates} define how values associated with qualified keys are modified.

\begin{definition}[\textbf{State update}]
  \label{def:state-update}
A \emph{state update} 
$\mSubst : \Addr \mapstopart (\Val \mapstopart \Val)$
is a function from qualified keys to values;
we denote with $\setenum{\bind{\amvA.\valK}{\valV}}$
the state update which maps $\amvA.\valK$ to~$\valV$.
We define $\keys{\mSubst}$ as the set of qualified keys $\amvA.\valK$ such that
$\amvA \in \dom{\mSubst}$ and $\valK \in \dom{\mSubst \amvA}$.
We apply updates to states as follows:
\[
(\bcSt \mSubst) \amvA = \bcKV[\amvA]
\quad \text{where} \quad
\bcKV[\amvA]  \valK =
\begin{cases}
  \mSubst \amvA \valK & \text{if $\amvA.\valK \in \keys{\mSubst}$} \\
  \bcSt \amvA \valK & \text{otherwise}
\end{cases}
\tag*{$\qedef$}
\]
\end{definition}

We denote with $\semCmd{\cmdC}{\bcSt,\bcEnv}{\amvA}$ the semantics of 
the statement $\cmdC$.
This semantics is either a blockchain state $\bcSti$, 
or it is undefined (denoted by $\bot$).
The semantics is parameterised over a state $\bcSt$, an address $\amvA$
(the contract wherein $\cmdC$ is evaluated),
and an \emph{environment} $\bcEnv : \Const \mapstopart \Val$,
used to evaluate the formal parameters
and the special names $\varSender$ and $\varValue$.
These names represent, respectively, the caller of the function,
and the amount of currency transferred along with the call.
We postulate that \varSender and \varValue 
are not used as formal parameters.

We define the auxiliary operators $+$ and $-$ on states as follows:
\[
\bcSt \circ (\amvA:\valN)
\; = \;
\bcSt\setenum{\bind{\amvA.\varBalance}{(\bcSt \amvA \varBalance) \,\circ\, \valN}}
\tag*{($\circ \in \setenum{+,-}$)}
\]
\ie, $\bcSt + \amvA:\valN$ updates $\bcSt$ 
by increasing the $\varBalance$ of $\amvA$ of $\valN$ currency units.

A \emph{transaction} $\txT$ is a term of the form:
\[
\ethtx{\valN}{\pmvA}{\cmvA}{\funF}{\vec{\valV}}
\]
Intuitively, $\pmvA$ is the address of the caller,
$\cmvA$ is the address of the called contract,
$\funF$ is the called function,
$n$ is the value transferred from $\pmvA$ to $\cmvA$, and
$\vec{\valV}$ is the sequence of actual parameters.
We denote the semantics of $\txT$ in $\bcSt$ as $\semTx{\txT}{\bcSt}$,
where the function $\semTx{\cdot}{\bcSt}$
is defined in~\Cref{fig:tins:tx-semantics},
which we briefly comment.

\begin{figure}[t]
  \[
  \hspace{-10pt}
  \begin{array}{c}
    \irule
    {\begin{array}{l}
       \contrFun{\funF}{\vec{\varX}}{\cmdC} \in \addrToContr{\cmvA}
       \\[4pt]
       \bcSt \, \pmvA \, \varBalance \geq \valN
       \\[4pt]
       \semCmd{\cmdC}{\bcSt-\pmvA:\valN+\cmvA:\valN,\, \setenum{\bind{\varSender}{\pmvA},\bind{\varValue}{\valN},\bind{\vec{\varX}}{\vec{\valV}}}}{\cmvA} = \bcSti
     \end{array}
    }
    {\semTx{\ethtx{\valN}{\pmvA}{\cmvA}{\funF}{\vec{\valV}}}{\bcSt} = \bcSti}
    \smallnrule{[Tx1]}
    \quad
    \irule
    {
    \begin{array}{l}
      \contrFun{\funF}{\vec{\varX}}{\cmdC} \in \addrToContr{\cmvA}
      \\[4pt]
      \bigg(
      \begin{array}{l}
        \!\bcSt \, \pmvA \, \varBalance < \valN \qquad \text{or}
        \\[4pt]
        \!\semCmd{\cmdC}{\bcSt-\pmvA:\valN+\cmvA:\valN,\, \setenum{\bind{\varSender}{\pmvA},\bind{\varValue}{\valN},\bind{\vec{\varX}}{\vec{\valV}}} }{\cmvA} = \bot \!\!
      \end{array}
      \bigg)
    \end{array}
    }
    {\semTx{\ethtx{\valN}{\pmvA}{\cmvA}{\funF}{\vec{\valV}}}{\bcSt} = \bcSt}
    \smallnrule{[Tx2]}
  \end{array}
  \]
  \vspace{-10pt}
  \caption{Semantics of transactions.}
  \label{fig:tins:tx-semantics}
\end{figure}

The semantics of a transaction $\txT = \ethtx{\valN}{\pmvA}{\cmvA}{\funF}{\vec{\valV}}$,
in a given blockchain state $\bcSt$, is a new state $\bcSti$.
Rule~\nrule{[Tx1]} handles the case where the transaction is successful:
this happens when $\pmvA$'s balance is at least $\valN$,
and the function call terminates in a non-error state.
Note that $\valN$ units of currency are transferred to $\cmvA$
\emph{before} starting to execute $\funF$,
and that the names $\varSender$ and $\varValue$ are bound, respectively,
to $\pmvA$ and $\valN$.
Rule~\nrule{[Tx2]} applies either when $\pmvA$'s balance is not enough,
or the execution of $\funF$ fails.
In these cases, $\txT$ does not alter the state.

A \emph{blockchain} $\bcB$ is a finite sequence of transactions;
we denote with $\bcEmpty$ the empty blockchain.
The semantics of a blockchain is obtained by folding 
the semantics of its transactions, starting from a given state $\sigma$:
\[
\semBc{\bcEmpty}{\bcSt} = \bcSt
\qquad
\semBc{\txT \bcB}{\bcSt} = \semBc{\bcB}{\scriptsize \semTx{\txT}{\bcSt}}
\]
Note that erroneous transactions can occur within a blockchain,
but they have no effect on its semantics (as rule \nrule{[Tx2]} makes them identities \wrt the append operation).
We assume that in the initial state of the blockchain,
denoted by $\bcStInit$,
each address $\amvA$ has a balance $\balanceInit{\amvA} \geq 0$,
while all the other keys are unbound.

We write $\semBc{\bcB}{}$ for $\semBc{\bcB}{\bcStInit}$,
where $\bcStInit \amvA = \setenum{\bind{\varBalance}{\balanceInit{\amvA}}}$.
We say that a state $\bcSt$ is \emph{reachable}
if $\bcSt = \semBc{\bcB}{}$ for some $\bcB$.

\begin{example}
  \label{ex:transactions}
  Consider the following functions of a contract at address $\cmvA$:
  \[
  \contrFun{\funF[0]}{}{\cmdAss{\valX}{1}}
  \qquad
  \contrFun{\funF[1]}{}{\cmdIfT{\valX=0}{\cmdSend{1}{\pmvB}}}
  \qquad
  \contrFun{\funF[2]}{}{\cmdSend{1}{\pmvB}}
  \]
  Let $\bcSt$ be a state such that
  $\bcSt \pmvA \varBalance \geq 2$,
  and let $\bcB = \txT[0]\txT[1]\txT[2]$, where:
  \[
  \txT[0] = \ethtx{0}{\pmvA}{\cmvA}{\funF[0]}{}
  \hspace{30pt}
  \txT[1] = \ethtx{1}{\pmvA}{\cmvA}{\funF[1]}{}
  \hspace{30pt}
  \txT[2] = \ethtx{1}{\pmvA}{\cmvA}{\funF[2]}{}
  \]
  By applying rule \nrule{[Tx1]} three times, we have that:
  \begin{align*}
    \semTx{\txT[0]}{\bcSt} 
    & = \semCmd{\cmdAss{\valX}{1}}{\bcSt,\,
      \setenum{\bind{\varSender}{\pmvA},\bind{\varValue}{0}}}{\cmvA} = \bcSt
      \setenum{\bind{\cmvA.\varX}{1}} = \bcSti
    \\
    \semTx{\txT[1]}{\bcSti} 
    & = \semCmd{\cmdIfT{\valX=0}{\cmdSend{1}{\pmvB}}}{\bcSti - \pmvA:1 + \cmvA:1,\,
      \setenum{\bind{\varSender}{\pmvA},\bind{\varValue}{1}}}{\cmvA} 
    \\
    & = \bcSti - \pmvA:1 + \cmvA:1 = \bcStii
    \\
    \semTx{\txT[2]}{\bcStii} 
    & = \semCmd{\cmdSend{1}{\pmvB}}{\bcStii - \pmvA:1 + \cmvA:1,\,
      \setenum{\bind{\varSender}{\pmvA},\bind{\varValue}{1}}}{\cmvA} = \bcStii - \pmvA:1 + \pmvB:1
  \end{align*}
  Summing up, 
  $\semBc{\bcB}{\bcSt} = \bcSt\setenum{\bind{\cmvA.\varX}{1}} - \pmvA:2 + \pmvB:1 + \cmvA:1$.
  \hfill\qedef
\end{example}

\section{Swapping transactions}
\label{sec:txswap}

We define two blockchain states to be \emph{observationally equivalent}
when they agree on the values associated to all the qualified keys.
Our formalisation is parameterised on a set of qualified keys $\QmvA$
over which we require the agreement.

\begin{definition}[\textbf{Observational equivalence}]
  \label{def:sim}
  For all $\QmvA \subseteq \QmvU$,
  we define 
  \mbox{$\bcSt \sim_{\QmvA} \bcSti$}
  iff 
  $\forall \qmvA \in \QmvA: \bcSt \qmvA = \bcSti \qmvA$.
  We say that $\bcSt$ and $\bcSti$ are \emph{observationally equivalent},
  in symbols $\bcSt \sim \bcSti$,
  when $\bcSt \sim_{\QmvA} \bcSti$ holds for all $\QmvA$.
  \hfill\qedef
\end{definition}

\begin{restatable}{lemma}{lemsimequiv}
  \label{lem:sim:equiv}
  For all $\QmvA,\QmvB \subseteq \QmvU$:
  \begin{inlinelist}
  \item \label{lem:sim:equiv:equiv}
    $\sim_{\QmvA}$ is an equivalence relation;
  \item \label{lem:sim:equiv:subseteq}
    if $\bcSt \sim_{\QmvA} \bcSti$ and $\QmvB \subseteq \QmvA$,
  then $\bcSt \sim_{\QmvB} \bcSti$;
  \item \label{lem:sim:equiv:sim}
    $\sim \, = \, \sim_{\QmvU}$.
  \end{inlinelist}
  \hfill\qedef
\end{restatable}

We extend the equivalence relations above to blockchains, 
by passing through their semantics.
For all $\QmvA$, we define 
$\bcB \sim_{\QmvA} \bcBi$
iff
$\semBc{\bcB}{\bcSt} \sim_{\QmvA} \semBc{\bcBi}{\bcSt}$
holds for all reachable $\bcSt$
(note that all the definitions and results in this paper apply to reachable states, 
since the unreachable ones do not represent actual contract executions).
We write $\bcB \sim \bcBi$ when $\bcB \sim_{\QmvA} \bcBi$ holds for all $\QmvA$.
The relation $\sim$ is a \emph{congruence} with respect to the append operation,
\ie if $\bcB \sim \bcBi$ then we can replace $\bcB$ with $\bcBi$ in a larger blockchain,
preserving its semantics.

\begin{restatable}{lemma}{lemsimappend}
  \label{lem:sim:append}
  $\bcB \sim \bcBi \implies \forall \bcB[0],\bcB[1] \, : \, \bcB[0] \bcB \bcB[1] \sim \bcB[0] \bcBi \bcB[1]$.
  \hfill\qedef
\end{restatable}

Two transactions are \emph{swappable} when
exchanging their order preserves observational equivalence.

\begin{definition}[\textbf{Swappability}]
  \label{def:swap}
  Two transactions $\txT \neq \txTi$ are \emph{swappable},
  in symbols $\txT \swap \txTi$, when $\txT \txTi \sim \txTi \txT$.
  \hfill\qedef
\end{definition}

\begin{example}
  \label{ex:swap}
  Recall the transactions in~\Cref{ex:transactions}.
  We have that $\txT[0] \swap \txT[2]$ and $\txT[1] \swap \txT[2]$,
  but $\txT[0] \nswap \txT[1]$
  \iftoggle{arxiv}
  {(see~\Cref{fig:swap:nocongruence} in~\Cref{sec:proofs}).}
  {(see Figure 5 in Appendix A of~\cite{BGM20arxiv}).}
  \hfill\qedef
\end{example}

We shall use the theory of trace languages originated from Mazurkiewicz's works~\cite{Mazurkiewicz88rex}
to study observational equivalence under various swapping relations.
Below, we fix the alphabet of trace languages
as the set $\Tx$ of all transactions.

\begin{definition}[\textbf{Mazurkiewicz equivalence}]
  \label{def:mazurkiewicz}
  Let $I$ be a symmetric and irreflexive relation on $\Tx$.
  The \emph{Mazurkiewicz equivalence} $\equivStSeq[I]$
  is the least congruence in the free monoid $\Tx^*$ such that:
  $\forall \txT, \txTi \in \Tx$:
  \(\;
  \txT \, I \, \txTi \implies \txT \txTi \equivStSeq[I] \txTi \txT
  \).
\end{definition}

\Cref{th:swap:mazurkiewicz} below states that the Mazurkiewicz equivalence
constructed on the swappability relation $\swap$ is an observational equivalence.
Therefore, we can transform a blockchain into an observationally equivalent one
by a finite number of exchanges of adjacent swappable transactions.

\begin{restatable}{theorem}{thswapmazurkiewicz}
  \label{th:swap:mazurkiewicz}
  $\equivStSeq[\swap] \;\; \subseteq \;\; \sim$.
  \hfill\qedef
\end{restatable}

\begin{example}
  \label{ex:swap:mazurkiewicz}
  We can rearrange the transactions in \Cref{ex:transactions} as
  $\txT[0] \txT[1] \txT[2] \sim \txT[0] \txT[2] \txT[1] \sim \txT[2] \txT[0] \txT[1]$.
  Instead, $\txT[1] \txT[0] \txT[2] \not\sim \txT[2] \txT[0] \txT[1]$
  (\eg, starting from a state $\bcSt$ such that $\bcSt\pmvA\varBalance = 2$ and $\bcSt\cmvA \valX = 0$, 
  \iftoggle{arxiv}
  {see~\Cref{fig:swap:bc-noncongruence} in~\Cref{sec:proofs}).}
  {see Figure 6 in Appendix A of~\cite{BGM20arxiv}).}
  \hfill\qedef
\end{example}

Note that the converse of~\Cref{th:swap:mazurkiewicz} does not hold:
indeed, $\bcB \equivStSeq[\swap] \bcBi$ requires that $\bcB$ and $\bcBi$
have the same length, while $\bcB \sim \bcBi$ may also hold for blockchains
of different length 
(\eg, $\bcBi = \bcB \txT$, where $\txT$ does not alter the state).

\paragraph*{Safe approximations of read/written keys}

Note that the relation $\swap$ is undecidable
whenever the contract language is Turing-equivalent.
So, to detect swappable transactions we follow a static approach,
consisting of two steps.
First, we over-approximate the set of keys read and written 
by transactions, by statically analysing the code of the called functions.
We then check a simple condition on these approximations (\Cref{def:pswapWR}),
to detect if two transactions can be swapped.
Since static analyses to over-approximate read and written variables
are quite standard~\cite{NNH99},
here we just rely on such approximations, by only assuming their safety.
In~\Cref{def:safeapprox} we state that a set $\QmvA$ safely
approximates the keys \emph{written} by $\txT$, 
when $\txT$ does not alter the state of the keys not in $\QmvA$.
Defining set of \emph{read} keys is a bit trickier: 
intuitively, we require that if we execute the transaction starting from two states that agree 
on the values of the keys in the read set, then these executions should be equivalent, 
in the sense that they do not introduce new differences between the resulting states 
(with respect to the difference already existing before).

\begin{definition}[\textbf{Safe approximation of read/written keys}]
  \label{def:safeapprox}
  Given a set of qualified keys $\QmvA$ and a transaction $\txT$, we define:
  \begin{align*}
    & \wapprox{\QmvA}{\txT}
    && \text{iff}
    && \forall \QmvB: \QmvB \cap \QmvA = \emptyset \implies \txT \sim_{\QmvB} \bcEmpty
    \\
    & \rapprox{\QmvA}{\txT}
    && \text{iff}
    && \forall \bcB,\bcBi,\QmvB: \bcB \sim_{\QmvA} \bcBi \land
       \bcB \sim_{\QmvB} \bcBi
       \implies \bcB\txT \sim_{\QmvB} \bcBi\txT
       \tag*{\qedef}
  \end{align*}
\end{definition}

\begin{example}
  \label{ex:safeapprox}
  Let \mbox{$\txT = \ethtx{1}{\pmvA}{\cmvA}{\funF}{}$},
  where $\contrFun{\funF}{}{\cmdSend{1}{\pmvB}}$ is a function of $\cmvA$.
  The execution of $\txT$ affects the $\varBalance$ of $\pmvA$, $\pmvB$ and $\cmvA$;
  however, $\cmvA.\varBalance$ is first incremented
  and then decremented, and so its value remains unchanged.
  Then,
  $\wapprox{\setenum{\expGet{\pmvA}{\varBalance},\expGet{\pmvB}{\varBalance}}}{\txT}$,
  and it is the smallest safe approximation of the keys written by $\txT$.
  To prove that $\rapprox{\QmvA = \setenum{\pmvA.\varBalance}}{\txT}$,
  assume two blockchains $\bcB$ and $\bcBi$ and a set of keys $Q$ such that
  $\bcB \sim_{\QmvA} \bcBi$ and $\bcB \sim_{Q} \bcBi$.
  If $\semBc{\bcB}{} \pmvA \varBalance < 1$,
  then by \nrule{[Tx2]}
  we have $\semBc{\bcB\txT}{} = \semBc{\bcB}{}$.
  Since $\bcB \sim_{\QmvA} \bcBi$, 
  then also $\semBc{\bcBi}{} \pmvA \varBalance < 1$,
  and so by \nrule{[Tx2]} we have
  $\semBc{\bcBi\txT}{} = \semBc{\bcBi}{}$.
  Then, 
  $\bcB\txT \sim_{Q} \bcBi \txT$.
  Otherwise, if $\semBc{\bcB}{} \pmvA \varBalance = n \geq 1$,
  then by \nrule{[Tx1]} the execution of $\txT$
  transfers one unit of currency from $\pmvA$ to $\pmvB$,
  so the execution of $\txT$ affects exactly 
  $\pmvA.\varBalance$ and $\pmvB.\varBalance$.
  So, it is enough to show that 
  $\bcB \sim_{\setenum{\qmvB}} \bcBi$ implies $\bcB \txT \sim_{\setenum{\qmvB}} \bcBi \txT$
  for $\qmvB \in \setenum{\pmvA.\varBalance, \pmvB.\varBalance}$.
  For $\qmvB = \pmvA.\varBalance$,
  we have that
  $\semBc{\bcBi\txT}{} \pmvA \varBalance = n - 1 = \semBc{\bcB\txT}{} \pmvA \varBalance$.
  For $\qmvB = \pmvB.\varBalance$,
  we have that
  $\semBc{\bcBi\txT}{} \pmvB \varBalance = \semBc{\bcBi}{} \pmvB \varBalance + 1 =
  \semBc{\bcB}{} \pmvB \varBalance + 1 = \semBc{\bcB\txT}{} \pmvB \varBalance$.
  Therefore, we conclude that $\rapprox{\QmvA}{\txT}$.
  \hfill\qedef
\end{example}

Widening a safe approximation (either of read or written keys)
preserves its safety;
further, the intersection
of two safe write approximations is still safe
\iftoggle{arxiv}
{(see \Cref{lem:safeapprox} in \Cref{sec:proofs}).}
{(see Lemma 6 in Appendix A of~\cite{BGM20arxiv}).}
From this, it follows that there exists a \emph{least} safe approximation
of the keys written by a transaction.

\paragraph*{Strong swappability}

We use safe approximations of the read/written keys 
to detect when two transactions are swappable.
To achieve that, we check whether two transactions 
$\txT$ and $\txTi$ operate on disjoint portions of the blockchain state.
More specifically, 
we recast in our setting Bernstein's conditions~\cite{bernstein}
for the parallel execution of processes:
it suffices to check that the set of keys written by $\txT$
is disjoint from those written or read by $\txTi$, and vice versa.
When this happens we say that the
two transactions are \emph{strongly swappable}.

\begin{definition}[\textbf{Strong swappability}]
  \label{def:pswap}
  We say that two transactions $\txT \neq \txTi$ are \emph{strongly swappable},
  in symbols $\txT \pswap{}{} \txTi$, when
  there exist $W,W',R,R' \subseteq \QmvU$ such that
  $\wapprox{W}{\txT}$,
  $\wapprox{W'}{\txTi}$,
  $\rapprox{R}{\txT}$,
  $\rapprox{R'}{\txTi}$,
  and:
  \[
    \big( R \cup W \big) \cap W'
    \; = \;
    \emptyset
    \; = \;
    \big( R' \cup W' \big) \cap W
    \tag*{\qedef}
  \]
\end{definition}

\begin{example}
  \label{ex:txswap:strong-swap}
  Let $\contrFun{\funF[1]}{}{\cmdSkip}$
  and $\contrFun{\funF[2]}{x}{\cmdSend{\varValue}{x}}$
  be functions of the contracts $\cmvA[1]$ and $\cmvA[2]$,
  respectively, 
  and consider the following transactions:
  \begin{align*}
    & \txT[1] = \ethtx{1}{\pmvA}{\cmvA[1]}{\funF[1]}{}
    && \txT[2] = \ethtx{1}{\pmvB}{\cmvA[2]}{\funF[2]}{\cmv{F}}
  \end{align*}
  where $\pmvA$, $\pmvB$, and $\cmv{F}$ are account addresses.
  To prove that $\txT[1] \pswap{}{} \txT[2]$, consider the following
  safe approximations of the written/read keys of $\txT[1]$ and
  $\txT[2]$, respectively:
  \begin{align*}
    & \wapprox{W_1 = \setenum{\pmvA.\varBalance, \cmvA[1].\varBalance}}{\txT[1]}
    && \rapprox{R_1 = \setenum{\pmvA.\varBalance}}{\txT[1]}
    \\
    & \wapprox{W_2 = \setenum{\pmvB.\varBalance, \cmv{F}.\varBalance}}{\txT[2]}
    && \rapprox{R_2 = \setenum{\pmvB.\varBalance}}{\txT[2]}
  \end{align*}
  Since $(W_1 \cup R_1) \cap W_2 = \emptyset = (W_2 \cup R_2) \cap W_1$, the two transactions are
  strongly swappable.
  Now, let:
  \[
  \txT[3] = \ethtx{1}{\pmvB}{\cmvA[2]}{\funF[2]}{\pmvA}
  \]
  and consider the following safe approximations $W_3$ and $R_3$:
  \begin{align*}
    & \wapprox{W_3 = \setenum{\pmvB.\varBalance, \pmvA.\varBalance}}{\txT[3]}
    \qquad
    \rapprox{R_3 = \setenum{\pmvB.\varBalance}}{\txT[3]}
  \end{align*}
  Since $W_1 \cap W_3 \ne \emptyset \ne W_2 \cap W_3$,
  then $\neg (\txT[1] \pswap{}{} \txT[3])$ and
  $\neg(\txT[2] \pswap{}{} \txT[3])$.
  \hfill\qedef
\end{example}

The following~\namecref{th:pswap-implies-swap} 
ensures the soundness of our approximation, \ie that
if two transactions are strongly swappable, then they are also swappable.
The converse implication does not hold,
as witnessed by~\Cref{ex:txswap:swap-notimplies-pswap}.

\begin{restatable}{theorem}{thpswapimpliesswap}
  \label{th:pswap-implies-swap}
  $\txT \pswap{}{} \txTi \implies \txT \swap \txTi$.
  \hfill\qedef
\end{restatable}

\begin{example}[Swappable transactions, not strongly]
  \label{ex:txswap:swap-notimplies-pswap}
  Consider the following functions and transactions 
  of a contract at address $\cmvA$:
  \begin{align*}
    & \contrFun{\funF[1]}{}{\cmdIfTE{\varSender = \pmvA \;\codeand\; \expLookup{\valK[1]} = 0}{\cmdAss{\valK[1]}{1}}{\cmdThrow}} 
    && \txT[1] = \ethtx{1}{\pmvA}{\cmvA}{\funF[1]}{}
    \\
    & \contrFun{\funF[2]}{}{\cmdIfTE{\varSender = \pmvB \;\codeand\; \expLookup{\valK[2]} = 0}{\cmdAss{\valK[2]}{1}}{\cmdThrow}}
    && \txT[2] = \ethtx{1}{\pmvB}{\cmvA}{\funF[2]}{}
  \end{align*}
  We prove that $\txT[1] \swap \txT[2]$.
  First, consider a state $\bcSt$ such that $\bcSt\pmvA\varBalance > 1$,
  $\bcSt\pmvB\varBalance > 1$,  $\bcSt\cmvA\varBalance = n$, $\bcSt\cmvA \valK[1] = 0$ and $\bcSt\cmvA \valK[2] = 0$.
  We have that:
  \[
  \semTx{\txT[1] \txT[2]}{\bcSt} = \bcSt\setenum{\bind{\cmvA.\valK[1]}{1},
    \bind{\cmvA.\valK[2]}{1},\bind{\cmvA.\varBalance}{n + 2}} = \semTx{\txT[2] \txT[1]}{\bcSt}
  \]
  In the second case, let $\bcSt$ be such that $\bcSt\pmvA\varBalance < 1$,
  or $\bcSt\pmvB\varBalance < 1$,
  or $\bcSt\cmvA \valK[1] \neq 0$,
  or $\bcSt\cmvA \valK[2] \neq 0$.
  It is not possible that the guards in $\funF[1]$ and $\funF[2]$
  are both true, so $\txT[1]$ or $\txT[2]$ raise an exception,
  leaving the state unaffected.
  Then, also in this case we have that
  \(
  \semTx{\txT[1] \txT[2]}{\bcSt} = \semTx{\txT[2] \txT[1]}{\bcSt}
  \),
  and so $\txT[1]$ and $\txT[2]$ are swappable.
  However, they are \emph{not} strongly swappable if
  there exist reachable states $\bcSt,\bcSti$
  such that $\bcSt \cmvA\valK[1] = 0 = \bcSti \cmvA\valK[2]$.
  To see why, let
  $W_1 = \setenum{\expGet{\pmvA}{\varBalance}, \expGet{\cmvA}{\varBalance},
    \expGet{\cmvA}{\valK[1]}}$.
  From the code of $\funF[0]$ we see that $W_1$ is the least
  safe over-approximation of the written
  keys of $\txT[1]$ ($\wapprox{W_1}{\txT[1]}$).
  This means that every safe approximation of $\txT[1]$ must include the keys
  of $W_1$.
  Similarly,
  $W_2 = \setenum{\expGet{\pmvB}{\varBalance},
    \expGet{\cmvA}{\varBalance}, \expGet{\cmvA}{\valK[2]}}$
  is the least safe over-approximation of the written keys of
  $\txT[2]$ ($\wapprox{W_2}{\txT[2]}$).
  Since the least safe approximations
  of the keys written by $\txT[1]$ and $\txT[2]$ are not disjoint,
  $\txT[1] \pswap{}{} \txT[2]$ does not hold.
  \hfill\qedef
\end{example}

\Cref{th:pswap:mazurkiewicz} states that
the Mazurkiewicz equivalence $\equivStSeq[\pswap{}{}]$
is stricter than $\equivStSeq[\swap]$.
Together with~\Cref{th:swap:mazurkiewicz},
if $\bcB$ is transformed into $\bcBi$
by exchanging adjacent strongly swappable transactions,
then $\bcB$ and $\bcBi$ are observationally equivalent.

\begin{restatable}{theorem}{thpswapmazurkiewicz}
  \label{th:pswap:mazurkiewicz}
  $\equivStSeq[\pswap{}{}] \; \subseteq \; \equivStSeq[\swap]$.
  \hfill\qedef
\end{restatable}

Note that if the contract language is Turing-equivalent, 
then finding approximations which satisfy 
the disjointness condition in~\Cref{def:pswap} is not computable, 
and so the relation $\pswap{}{}$ is undecidable.

\paragraph*{Parameterised strong swappability}

Strongly swappability abstracts from the actual static analysis used to obtain
the safe approximations:
it is sufficient that such an analysis exists.
\Cref{def:pswapWR} below parameterises strong swappability over 
a static analysis, which we represent as
a function from transactions to sets of qualified keys,
just requiring it to be a safe approximation.
Formally, we say that $\wset{}$ is a \emph{static analysis of written keys}
when $\wapprox{\wset{\txT}}{\txT}$, for all $\txT$;
similarly, $\rset{}$ is a \emph{static analysis of read keys}
when $\rapprox{\rset{\txT}}{\txT}$, for all $\txT$.

\begin{definition}[\textbf{Parameterised strong swappability}]
  \label{def:pswapWR}
  Let $\wset{}$ and $\rset{}$ be static analyses of written/read keys.
  We say that $\txT$, $\txTi$ are
  \emph{strongly swappable \wrt $\wset{}$ and $\rset{}$},
  in symbols $\txT \pswapWR \txTi$, if:
  \[
  \big( \rset{\txT} \cup \wset{\txT} \big) \cap \wset{\txTi}
  \; = \;
  \emptyset
  \; = \;
  \big( \rset{\txTi} \cup \wset{\txTi} \big) \cap \wset{\txT}
  \tag*{\qedef}
  \]
\end{definition}

Note that an effective procedure for computing $\wset{}$ and $\rset{}$ gives
an effective procedure to determine whether two transactions are (strongly) swappable.

\begin{restatable}{lemma}{lempswapWRimpliespswap}
  \label{lem:pswapWR-implies-pswap}
  For all static analyses $\wset{}$ and $\rset{}$:
  \begin{inlinelist}
  \item \label{lem:pswapWR-implies-pswap:subseteq}
    $\pswapWR \subseteq \pswap{}{}$;
  \item \label{lem:pswapWR-implies-pswap:decidable}
    if $\wset{}$ and $\rset{}$ are computable, then $\pswapWR$ is decidable.
  \end{inlinelist}
  \hfill\qedef
\end{restatable}

From the inclusion in item~\ref{lem:pswapWR-implies-pswap:subseteq}
of \Cref{lem:pswapWR-implies-pswap} and from~\Cref{th:pswap:mazurkiewicz}
we obtain:

\begin{restatable}{theorem}{thpswapWRmazurkiewicz}
  \label{th:pswapWR:mazurkiewicz}
  $\equivStSeq[\pswapWR] \; \subseteq \; \equivStSeq[\pswap{}{}] \; \subseteq \; \equivStSeq[\swap]$.
  \hfill\qedef
\end{restatable}


\section{True concurrency for blockchains}
\label{sec:txpar}

Given a swappability relation $\relR$, 
we transform a sequence of transactions $\bcB$ into an 
\emph{occurrence net} $\PNet{\relR}{\bcB}$,
which describes the partial order induced by $\relR$. 
Any concurrent execution of the transactions in $\bcB$
which respects this partial order is equivalent 
to the serial execution of $\bcB$ (\Cref{th:bc-to-pnet}).

\paragraph{From blockchains to occurrence nets}

We start by recapping the notion of Petri net~\cite{Reisig85book}.
A \emph{Petri net} is a tuple
$\netN = (\Places, \Transitions, \Arcs, \markM[0])$,
where $\Places$ is a set of \emph{places},
$\Transitions$ is a set of \emph{transitions}
(with $\Places \cap \Transitions = \emptyset$),
and
$\Arcs : (\Places \times \Transitions) \cup (\Transitions \times \Places) \rightarrow \Nat$
is a \emph{weight function}.
The state of a net is a \emph{marking}, \ie a multiset
$\markM: \Places \rightarrow \Nat$ defining how many \emph{tokens}
are contained in each place;
we denote with $\markM[0]$ the initial marking.
The behaviour of a Petri net is specified as 
a transition relation between markings:
intuitively, a transition $\trT$ is enabled at $\markM$ 
when each place $\placeP$ 
has at least $\Arcs(\placeP,\trT)$ tokens in $\markM$.
When an enabled transition $\trT$ is fired, it consumes 
$\Arcs(\placeP,\trT)$ tokens from each $\placeP$,
and produces $\Arcs(\trT,\placePi)$ tokens in each $\placePi$.
Formally, given $x \in \Places \cup \Transitions$, 
we define the \emph{preset} $\pre{x}$ and the \emph{postset} $\post{x}$ 
as multisets:
$\pre{x}(y) = \Arcs(y,x)$, and 
$\post{x}(y) = \Arcs(x,y)$.
A transition $\trT$ is \emph{enabled} at $\markM$
when $\pre{\trT} \subseteq \markM$.
The transition relation between markings is defined as
$\markM \trans{\trT} \markM - \pre{\trT} + \post{\trT}$, 
where $\trT$ is enabled.
We say that $\trT[1] \cdots \trT[n]$ is a 
\emph{firing sequence from $\markM$ to $\markMi$} when
$\markM \trans{\trT[1]} \cdots \trans{\trT[n]} \markMi$,
and in this case we say that
$\markMi$ is \emph{reachable from} $\markM$.
We say that $\markMi$ is \emph{reachable} when it is reachable from $\markM[0]$.

An \emph{occurrence net}~\cite{Best87tcs} is a Petri net such that:
\begin{inlinelist}
\item \label{item:petri:onet:p}
$\card{\post{\plP}} \leq 1$ for all $\plP$;
\item \label{item:petri:onet:mzero}
$\card{\pre{\plP}} = 1$ if $\plP \not\in \markM[0]$, and
$\card{\pre{\plP}} = 0$ if $\plP\in \markM[0]$;
\item \label{item:petri:onet:relation}
$\Arcs$ is a relation, \ie $\Arcs(x,y) \leq 1$ for all $x,y$;
\item \label{item:petri:onet:po}
$\Arcs^*$ is a acyclic, \ie
$\forall x,y \in \Places \cup \Transitions : (x,y) \in \Arcs^* \land (y,x) \in \Arcs^* \implies x=y$
(where $\Arcs^*$ is the reflexive and transitive closure of $\Arcs$).
\end{inlinelist}

In~\Cref{def:bc-to-pnet} we transform 
a blockchain $\bcB = \txT[1] \cdots \txT[n]$ 
into a Petri net $\PNet{\relR}{\bcB}$, 
where $\relR$ is an arbitrary relation between transactions.
Although any relation $\relR$ ensures that $\PNet{\relR}{\bcB}$ 
is an occurrence net (\Cref{lem:bc-to-pnet:onet} below),
our main results hold when $\relR$ is a strong swappability relation.
The transformation works as follows:
the $i$-th transaction in $\bcB$ is rendered as a transition $(\txT[i],i)$ in $\PNet{\relR}{\bcB}$, and
transactions related by $\relR$ are transformed into concurrent transitions.
Technically, this concurrency is specified as 
a relation $<$ between transitions, such that
$(\txT[i],i) < (\txT[j],j)$ whenever $i < j$, 
but $\txT[i]$ and $\txT[j]$ are not related by $\relR$.
The places, the weight function, and the initial marking of $\PNet{\relR}{\bcB}$ 
are chosen to ensure that the firing ot transitions respects the relation $<$.

\begin{figure}[t]
  \begin{align*}
    & \Transitions 
    = \setcomp{(\txT[i],i)}{1 \leq i \leq n}
    \\[5pt] 
    & \Places 
    \; = \setcomp{(*,\trT)}{\trT \in \Transitions} \cup 
      \setcomp{(\trT,*)}{\trT \in \Transitions} \cup \setcomp{(\trT,\trTi)}{\trT < \trTi}
    \\
    & \hspace{22pt} \text{where }
      (\txT,i) < (\txTi,j) \eqdef (i < j) \,\land\, \neg(\txT \,\relR\, \txTi)
    \\[5pt]
    & \Arcs(x,y) 
    = \begin{cases}
      1 & \text{if $y=\trT$ and \big($x=(*,\trT)$ or $x=(\trTi,\trT)$\big)} \\
      1 & \text{if $x=\trT$ and \big($y=(\trT,*)$ or $y =(\trT,\trTi)$\big)} \\
      0 & \text{otherwise}
    \end{cases}
    \hspace{10pt}
    \markM[0](\plP)
    = \begin{cases}
      1 & \text{if $\plP = (*,\trT)$} \\
      0 & \text{otherwise}
    \end{cases}
  \end{align*}
  \vspace{-10pt}
  \caption{Construction of a Petri net from a blockchain $\bcB = \txT[1] \cdots \txT[n]$.}
  \label{def:bc-to-pnet}
\end{figure}

\begin{example}
  \label{ex:petri:1}
  Consider the following transactions and functions of a contract $\cmvA$:
  \begin{align*}
    & \txT[\funF] = \ethtx{0}{\pmvA}{\cmvA}{\funF}{}
    &&  \contrFunSig{\funF}{} \, \{
       \cmdIfTE{\expLookup{\valX} = 0}{\cmdAss{\valY}{1}}{\cmdThrow} \}
    \\
    & \txT[\funG] = \ethtx{0}{\pmvA}{\cmvA}{\funG}{}
    && \contrFunSig{\funG}{} \, \{
    \cmdIfTE{\expLookup{\valY} = 0}{\cmdAss{\valX}{1}}{\cmdThrow} \}
    \\
    & \txT[\funH] = \ethtx{0}{\pmvA}{\cmvA}{\funH}{}
    && \contrFunSig{\funH}{} \, \{
    \cmdAss{\valZ}{1} \}
  \end{align*}
  Let
  \(
  \QmvA[\funF]^{w} = \QmvA[\funG]^{r} = \setenum{\cmvA.\valY},
  \QmvA[\funF]^{r} = \QmvA[\funG]^{w} = \setenum{\cmvA.\valX},
  \QmvA[\funH]^{w} = \setenum{\cmvA.\valZ},
  \QmvA[\funH]^{r} = \emptyset
  \). 
  It is easy to check that 
  these sets are safe approximations of their transactions
  (\eg, $\QmvA[\funF]^{w}$ safely approximates the keys written by $\txT[\funF]$).
  By~\Cref{def:pswap} we have that
  $\txT[\funF] \pswap{}{} \txT[\funH]$,
  $\txT[\funG] \pswap{}{} \txT[\funH]$,
  but $\neg(\txT[\funF] \pswap{}{} \txT[\funG])$.
  We display $\PNet{\pswap{}{}}{\txT[\funF]\txT[\funH]\txT[\funG]}$ 
  in~\Cref{fig:petri:1},
  where $\trT[\funF] = (\txT[\funF],1)$, 
  $\trT[\funH] = (\txT[\funH],2)$, and
  $\trT[\funG] = (\txT[\funG],3)$.
  Note that $\trT[\funG]$ can only be fired after $\trT[\funF]$,
  while $\trT[\funH]$ can be fired independently from 
  $\trT[\funF]$ and $\trT[\funG]$.
  This is coherent with the fact that $\txT[\funH]$ is swappable with
  both $\txT[\funF]$ and $\txT[\funG]$, while $\txT[\funF]$ and $\txT[\funG]$
  are not swappable.
  \hfill\qedef
\end{example}

\begin{figure}[t]
  \centering
  \begin{tikzpicture}[>=stealth',scale=0.7] 
    \tikzstyle{place}=[circle,thick,draw=black!75,minimum size=5mm]
    \tikzstyle{transition}=[rectangle,thick,draw=black!75,minimum size=5mm]
    \tikzstyle{edge}=[->,thick,draw=black!75]
    \tikzstyle{edgered}=[->,thick,draw=red!75]
    \tikzstyle{edgeblu}=[->,thick,draw=blue!75]
    \node[place] (p*1) at (0,3) [label = above:{$(*,\trT[\funF])$}] {};
    \node[place] (p*2) at (11,3) [label = above:{$(*,\trT[\funH])$}] {};
    \node[place] (p*3) at (6,4.5) [label = right:{$(*,\trT[\funG])$}] {};
    \node[place] (p1*) at (2,4.5) [label = right:{$(\trT[\funF],*)$}] {};
    \node[place] (p2*) at (15,3) [label = above:{$(\trT[\funH],*)$}] {};
    \node[place] (p3*) at (8,3) [label = above:{$(\trT[\funG],*)$}] {};
    \node[place] (p13) at (4,3) [label = above:{$(\trT[\funF],\trT[\funG])$}] {};
    \node[place,token = 1] (tok1) at (0,3) {};
    \node[place,token = 1] (tok2) at (11,3) {};
    \node[place,token = 1] (tok3) at (6,4.5) {};
    \node[transition] (t1) at (2,3) {$\trT[\funF]$}
    edge[pre] node {} (p*1)
    edge[post] node {} (p1*)
    edge[post] node {} (p13);
    \node[transition] (t2) at (13,3) {$\trT[\funH]$}
    edge[pre] node {} (p*2)
    edge[post] node {} (p2*);
    \node[transition] (t3) at (6,3) {$\trT[\funG]$}
    edge[pre] node {} (p*3)
    edge[pre] node {} (p13)
    edge[post] node {} (p3*);
  \end{tikzpicture}
  \caption{Occurrence net for \Cref{ex:petri:1}.}
  \label{fig:petri:1}
\end{figure}

\begin{restatable}{lemma}{lembctopnetonet}
\label{lem:bc-to-pnet:onet}
  $\PNet{\relR}{\bcB}$ is an occurrence net, for all $\relR$ and $\bcB$.
\end{restatable}

\paragraph{Step firing sequences}

\Cref{th:bc-to-pnet} below establishes a correspondence between 
concurrent and serial execution of transactions.
Since the semantics of serial executions is given in terms of blockchain states $\bcSt$, 
to formalise this correspondence we use the same semantics domain
also for concurrent executions.
This is obtained in two steps.
First, we define concurrent executions of $\bcB$  
as the \emph{step firing sequences} (\ie finite sequences of \emph{sets} of transitions) 
of $\PNet{\pswap{}{}}{\bcB}$.
Then, we give a semantics to step firing sequences, in terms of blockchain states.

We denote finite sets of transitions, called \emph{steps},
as $\trSU,\trSUi,\hdots$.
Their preset and postset are defined as
\(
\textstyle
\pre{\trSU} = \sum_{\plP \in \trSU} \pre{\plP} 
\)
and
\(
\textstyle
\post{\trSU} = \sum_{\plP \in \trSU} \post{\plP}
\),
respectively.
We say that $\trSU$ is \emph{enabled at $\markM$} 
when $\pre{\trSU} \leq \markM$, and in this case
firing $\trSU$ results in the move 
$\markM \trans{\trSU} \markM - \pre{\trSU} + \post{\trSU}$.
Let $\vec{\trSU} = \trSU[1] \cdots \trSU[n]$ be a finite sequence of steps.
We say that $\vec{\trSU}$ is a 
\emph{step firing sequence from $\markM$ to $\markMi$}
if $\markM \trans{\trSU[1]} \cdots \trans{\trSU[n]} \markMi$, 
and in this case we write
$\markM \xrightarrow{\vec{\trSU}} \markMi$.

\paragraph{Concurrent execution of transactions}

We now define how to execute transactions in parallel. 
The idea is to execute transactions in \emph{isolation},
and then merge their changes, whenever they are mutually disjoint.
The state updated resulting from the execution of a transaction
are formalised as in \Cref{def:state-update}.

An \emph{update collector} is a function $\WR{}$ that,
given a state $\bcSt$ and a transaction $\txT$,
gives an update $\mSubst = \WR{\bcSt,\txT}$ which maps (at least)
the updated qualified keys to their new values.
In practice, update collectors can be obtained by instrumenting
the run-time environment of smart contracts, 
so to record the state changes resulting from the execution of transactions.
We formalise update collectors abstracting from the 
implementation details of such an instrumentation:

\begin{definition}[\textbf{Update collector}]
  \label{def:wr}
  We say that a function $\WR{}$ is an \emph{update collector} when
  $\semTx{\txT}{\bcSt} = \bcSt (\WR{\bcSt,\txT})$,
  for all $\bcSt$ and $\txT$.
  \hfill\qedef
\end{definition}

There exists a natural ordering of collectors,
which extends the ordering between state updates
(\ie, set inclusion, when interpreting them as sets of substitutions):
namely, $\WR{} \sqsubseteq \WRi{}$ holds when
$\forall \bcSt, \txT : \WR{\bcSt,\txT} \subseteq \WRi{\bcSt,\txT}$.
The following~\namecref{lem:wr:minimal} characterizes the least 
update collector \wrt this ordering.

\begin{lemma}[\textbf{Least update collector}]
  \label{lem:wr:minimal}
  Let \mbox{$\WRmin{\bcSt,\txT} = \semTx{\txT}{\bcSt} - \bcSt$},
  where we define $\bcSti - \bcSt$ as
  $\bigcup_{{\bcSti \qmvA \neq \bcSt \qmvA}} \setenum{\bind{\qmvA}{\bcSti \qmvA}}$.
  Then, $\WRmin{}$ is the least update collector.
  \hfill\qedef
\end{lemma}

The merge of two state updates is the union of the corresponding substitutions;
to avoid collisions, we make the merge operator undefined 
when the domains of the two updates overlap.

\begin{definition}[\textbf{Merge of state updates}]
  \label{def:merge}
  Let $\mSubst[0]$, $\mSubst[1]$ be state updates.
  When $\keys{\mSubst[0]} \cap \keys{\mSubst[1]} = \emptyset$,
  we define $\mSubst[0] \mrg \mSubst[1]$ as follows:
  \[
  (\mSubst[0] \mrg \mSubst[1]) \qmvA =
  \begin{cases}
    \mSubst[0] \qmvA & \text{if $\qmvA \in \keys{\mSubst[0]}$} \\
    \mSubst[1] \qmvA & \text{if $\qmvA \in \keys{\mSubst[1]}$} \\
    \bot & \text{otherwise}
  \end{cases}
  \tag*{\qedef}
  \]
\end{definition}

The merge operator enjoys the commutative monoidal laws, and can therefore be
extended to (finite) sets of state updates.

We now associate step firing sequences with state updates.
The semantics of a step $\trSU = \setenum{(\txT[1],1),\ldots,(\txT[n],n)}$
in $\bcSt$ is obtained by applying to $\bcSt$ 
the merge of the updates $\WR{\bcSt,\txT[i]}$, 
for all $i \in 1..n$ --- whenever the merge is defined.
The semantics of a step firing sequence is then obtained by
folding that of its steps.

\begin{definition}[\textbf{Semantics of step firing sequences}]
  We define the semantics of step firing sequences,
  given $\WR{}$ and $\bcSt$, as:
  \[
  \msem{\WR{}}{\bcSt}{\bcEmpty} 
  \; = \;
  \bcSt
  \hspace{30pt}
  \msem{\WR{}}{\bcSt}{\trSU \vec{\trSU}} 
  \; = \;
  \msem{\WR{}}{\bcSti}{\vec{\trSU}} 
  \quad \text{where } \bcSti = 
  \msem{\WR{}}{\bcSt}{\trSU}
  =
  \bcSt \bigoplus_{(\txT,i) \in \trSU} \WR{\bcSt,\txT}
  \tag*{\qedef}
  \]
\end{definition}

\begin{example}
  \label{ex:txpar:txs-semantics}
  Let $\trT[\funF], \trT[\funG]$, and $\trT[\funH]$ be as in~\Cref{ex:petri:1},
  and let $\bcSt \cmvA x = \bcSt \cmvA y = 0$.
  Since
  $\WRmin{\bcSt,\txT[\funF]} = \setenum{\bind{\cmvA.y}{1}}$,
  $\WRmin{\bcSt,\txT[\funG]} = \setenum{\bind{\cmvA.x}{1}}$, and
  $\WRmin{\bcSt,\txT[\funH]} = \setenum{\bind{\cmvA.z}{1}}$,
  we have:
  \begin{align*}
    & \msem{\WRmin{}}{\bcSt}{\setenum{\trT[\funF],\trT[\funH]}} 
      = \bcSt(\setenum{\bind{\cmvA.y}{1}} \mrg \setenum{\bind{\cmvA.z}{1}}) 
      = \bcSt\setenum{\bind{\cmvA.y}{1},\bind{\cmvA.z}{1}}
    \\
    & \msem{\WRmin{}}{\bcSt}{\setenum{\trT[\funG],\trT[\funH]}} 
      = \bcSt(\setenum{\bind{\cmvA.x}{1}} \mrg \setenum{\bind{\cmvA.z}{1}}) 
      = \bcSt\setenum{\bind{\cmvA.x}{1},\bind{\cmvA.z}{1}}
    \\
    & \msem{\WRmin{}}{\bcSt}{\setenum{\trT[\funF],\trT[\funG]}} 
      = (\bcSt\setenum{\bind{\cmvA.y}{1}} \mrg \setenum{\bind{\cmvA.x}{1}})
      = \bcSt\setenum{\bind{\cmvA.y}{1},\bind{\cmvA.x}{1}}
  \end{align*}
  Note that, for all $\bcSt$:
  \[
  \begin{array}{c}
    \semBc{\txT[\funF]\txT[\funH]}{\bcSt} 
    \; = \; \semBc{\txT[\funH]\txT[\funF]}{\bcSt}
    \; = \; \bcSt\setenum{\bind{\cmvA.y}{1},\bind{\cmvA.z}{1}} 
    \; = \; \msem{\WRmin{}}{\bcSt}{\setenum{\trT[\funF],\trT[\funH]}} 
    \\
    \semBc{\txT[\funG]\txT[\funH]}{\bcSt} 
    \; = \; \semBc{\txT[\funH]\txT[\funG]}{\bcSt}
    \; = \; \bcSt\setenum{\bind{\cmvA.x}{1},\bind{\cmvA.z}{1}} 
    \; = \;
    \msem{\WRmin{}}{\bcSt}{\setenum{\trT[\funG],\trT[\funH]}}
  \end{array}
  \]
  So, the serial execution of $\txT[\funF]$ and $\txT[\funH]$ 
  (in both orders) is equal to their concurrent execution
  (similarly  for $\txT[\funG]$ and $\txT[\funH]$).
  Instead, for all $\bcSt$ such that $\bcSt \cmvA \valX = \bcSt \cmvA \valY = 0$: 
  \[
    \semBc{\txT[\funF]\txT[\funG]}{\bcSt} = 
    \bcSt\setenum{\bind{\cmvA.y}{1}}
    \qquad
    \semBc{\txT[\funG]\txT[\funF]}{\bcSt} = 
    \bcSt\setenum{\bind{\cmvA.x}{1}}
    \qquad
    \msem{\WRmin{}}{\bcSt}{\setenum{\trT[\funF],\trT[\funG]}} =
    \bcSt\setenum{\bind{\cmvA.y}{1},\bind{\cmvA.x}{1}}
  \]
  So, concurrent executions of $\txT[\funF]$ and $\txT[\funG]$ may 
  differ from serial ones.
  This is coherent with the fact that, in~\Cref{fig:petri:1},
  $\trT[\funF]$ and $\trT[\funG]$ are \emph{not} concurrent.
  \hfill\qedef
\end{example}

\paragraph{Concurrent execution of blockchains}

\Cref{th:bc-to-pnet} relates serial executions of transactions 
to concurrent ones (which are rendered as step firing sequences).
Item~\ref{th:bc-to-pnet:confluence} establishes a confluence property:
if two step firing sequences lead to the same marking,
then they also lead to the same blockchain state.
Item~\ref{th:bc-to-pnet:bc} ensures that the blockchain,
interpreted as a sequence of transitions, is a step firing sequence,
and it is \emph{maximal} 
(\ie, there is a bijection between the transactions in the blockchain 
and the transitions of the corresponding net).
Finally, item~\ref{th:bc-to-pnet:maximal}
ensures that executing maximal step firing sequences 
is equivalent to executing serially the blockchain.

\begin{restatable}{theorem}{thbctopnet}
  \label{th:bc-to-pnet}
  Let $\bcB = \txT[1] \cdots \txT[n]$. 
  Then, in $\PNet{\pswap{}{}}{\bcB}$:
  \begin{enumerate}[(a)]
    \vspace{-5pt}
  \item \label{th:bc-to-pnet:confluence}
    if $\markM[0] \xrightarrow{\vec{\trSU}} \markM$ and
    $\markM[0] \xrightarrow{\vec{\trSUi}} \markM$,
    then
    $\msem{\WRmin{}}{\bcSt}{\vec{\trSU}} = \msem{\WRmin{}}{\bcSt}{\vec{\trSUi}}$,
    for all reachable $\bcSt$;
    
  \item \label{th:bc-to-pnet:bc}
    $\setenum{(\txT[1],1)} \cdots \setenum{(\txT[n],n)}$
    is a maximal step firing sequence;

  \item \label{th:bc-to-pnet:maximal}
    for all maximal step firing sequences $\vec{\trSU}$,
    for all reachable $\bcSt$,
    $\msem{\WRmin{}}{\bcSt}{\vec{\trSU}} = \semBc{\bcB}{\bcSt}$.

\end{enumerate}
\end{restatable}

Remarkably, the implications of \Cref{th:bc-to-pnet} also apply to $\PNet{\pswapWR}{\bcB}$.

\begin{example} \label{ex:petri:2}
  Recall $\bcB = \txT[\funF]\txT[\funH]\txT[\funG]$ 
  and $\PNet{\pswap{}{}}{\bcB}$ from~\Cref{ex:petri:1}, 
  let $\vec{\trSU} = \setenum{\trT[\funF],\trT[\funH]}\setenum{\trT[\funG]}$, and
  let $\bcSt$ be such that
  $\bcSt \cmvA x = \bcSt \cmvA y = 0$.
  As predicted by item~\ref{th:bc-to-pnet:maximal} of~\Cref{th:bc-to-pnet}:
  \[
  \semBc{\bcB}{\bcSt} 
  \; = \; 
  \bcSt \setenum{\bind{\cmvA.y}{1}} \setenum{\bind{\cmvA.z}{1}}
  \; = \; 
  \msem{\WRmin{}}{\bcSt}{\vec{\trSU}}
  \]
  Let $\vec{\trSUi} = \setenum{\trT[\funF]}\setenum{\trT[\funG],\trT[\funH]}$.
  We have that $\vec{\trSU}$ and $\vec{\trSUi}$ lead to the same marking,
  where the places $(\trT[\funF],*)$, $(\trT[\funG],*)$ and $(\trT[\funH],*)$
  contain one token each, while the other places have no tokens.
  By item~\ref{th:bc-to-pnet:confluence} of \Cref{th:bc-to-pnet}
  we conclude that 
  $\msem{\WRmin{}}{\bcSt}{\vec{\trSU}} = \msem{\WRmin{}}{\bcSt}{\vec{\trSUi}}$.
  Now, let $\vec{\trSUii} = \setenum{\trT[\funH]}\setenum{\trT[\funF],\trT[\funG]}$.
  Note that, although $\vec{\trSUii}$ is maximal, it is not a step firing sequence,
  since the second step is not enabled
  (actually, $\trT[\funF]$ and $\trT[\funG]$ are not concurrent, as pointed out 
  in~\Cref{ex:txpar:txs-semantics}).
  Therefore, the items of~\Cref{th:bc-to-pnet} do not apply to $\vec{\trSUii}$,
  coherently with the fact that $\vec{\trSUii}$ does not represent
  any sequential execution of $\bcB$.
  \hfill\qedef
\end{example}

\section{Case study: ERC-721 token}
\label{sec:erc721}

We now apply our theory to an archetypal Ethereum smart contract,
which implements a ``non-fungible token'' following the standard 
ERC-721 interface~\cite{ERC721,Frowis19fc}. 
This contract defines the functions to transfer tokens between users,
and to delegate their trade to other users.
Currently, token transfers involve $\sim{50}\%$ of the 
transactions on the Ethereum blockchain~\cite{tokens},
with larger peaks due to popular contracts like Cryptokitties~\cite{Young17cointelegraph}.

We sketch below the implementation of the \code{Token} contract, 
using Solidity, the main high-level smart contract language in Ethereum 
(see
\iftoggle{arxiv}
{\Cref{sec:full-erc721}}
{Appendix B of~\cite{BGM20arxiv}}
for the full implementation).

The contract state is defined by the following mappings:
\begin{lstlisting}[language=solidity,numbersep=10pt]
mapping(uint256 => address) owner;
mapping(uint256 => bool) exists;
mapping(address => uint256) balance;
mapping (address => mapping (address => bool)) operatorApprovals;
\end{lstlisting}

Each token is uniquely identified by an integer value (of type \code{uint256}),
while users are identified by an \code{address}.
The mapping \code{owner} maps tokens to their owners' addresses 
(the zero address is used to denote a dummy owner).
The mapping \code{exists} tells whether a token has been created or not,
while \code{balance} gives the number of tokens owned by each user.
The mapping \code{operatorApprovals} allows a user to delegate the transfer of all her tokens
to third parties.

The function \code{transferFrom} transfers a token from the owner to another user.
The \code{require} assertion rules out some undesirable cases,
\eg, if the token does not exist, 
or it is not owned by the \code{from} user,
or the user attempts to transfer the token to himself.
Once all these checks are passed, the transfer succeeds if 
the \code{sender} of the transaction owns the token, or if he has been delegated by the owner.
The mappings \code{owner} and \code{balance} are updated as expected.

\begin{lstlisting}[language=solidity,numbersep=10pt]
function transferFrom(address from, address to, uint256 id) external {
  require (exists[id] && from==owner[id] 
          && from!=to && to!=address(0));      
  if (from==msg.sender || operatorApprovals[from][msg.sender]) {
    owner[id] = to;
    balance[from] -= 1;
    balance[to] += 1;
  }
}
\end{lstlisting}

The function \code{setApprovalForAll} delegates the transfers of all the tokens of the \code{sender} to the \code{operator} when the boolean \code{isApproved} is true, 
otherwise it revokes the delegation.

\begin{lstlisting}[language=solidity,numbersep=10pt]
function setApprovalForAll(address operator, bool isApproved) external {
  operatorApprovals[msg.sender][operator] = isApproved;
}
\end{lstlisting}

Assume that user $\pmvA$ owns two tokens, identified by the integers $1$ and $2$,
and consider the following transactions:
\begin{align*}
  & \txT[1] = \ethtx{0}{\pmvA}{\codeAddr{Token}}
    {\codeFun{transferFrom}}{\pmvA,\pmv{P},1}
  \\
  & \txT[2] = \ethtx{0}{\pmvA}{\codeAddr{Token}}
    {\codeFun{setApprovalForAll}}{\pmvB,\true} 
  \\
  & \txT[3] = \ethtx{0}{\pmvB}{\codeAddr{Token}}
    {\codeFun{transferFrom}}{\pmvA,\pmv{Q},2}
  \\
  & \txT[4] = \ethtx{0}{\pmv{P}}{\codeAddr{Token}}
    {\codeFun{transferFrom}}{\pmv{P},\pmvB,1}
\end{align*}

We have that
$\txT[1] \pswap{}{} \txT[2]$, 
$\txT[2] \pswap{}{} \txT[4]$, and 
$\txT[3] \pswap{}{} \txT[4]$
(this can be proved \eg by using the static approximations in
\iftoggle{arxiv}
{\Cref{sec:full-erc721})}
{Appendix B of~\cite{BGM20arxiv})},
while the other combinations are not swappable.
%
Let $\bcB = \txT[1] \txT[2] \txT[3] \txT[4]$.
The resulting occurrence net is displayed in~\Cref{fig:erc721:petri}.
For instance, 
let $\vec{\trSU} = \setenum{\txT[1],\txT[2]} \setenum{\txT[3],\txT[4]}$,
\ie $\txT[1]$ and $\txT[2]$ are executed concurrently,
as well as $\txT[3]$ and $\txT[4]$.
From item~\ref{th:bc-to-pnet:maximal} of \Cref{th:bc-to-pnet}
we have that this concurrent execution is equivalent to the serial one.

Although this example deals with the marginal case where 
the sender and the receiver of tokens overlap, 
in practice
the large majority of transactions in a block either involves distinct users, 
or invokes distinct ERC-721 interfaces, 
making it possible to increase the degree of concurrency 
of \code{transferFrom} transactions.

\begin{figure}[t]
  \centering
  \begin{tikzpicture}[>=stealth',scale=0.65] 
    \tikzstyle{place}=[circle,thick,draw=black!75,minimum size=3mm]
    \tikzstyle{transition}=[rectangle,thick,draw=black!75,minimum size=5mm]
    \tikzstyle{edge}=[->,thick,draw=black!75]
    \tikzstyle{edgered}=[->,thick,draw=red!75]
    \tikzstyle{edgeblu}=[->,thick,draw=blue!75]
    \node[place] (p*1) at (0,3) [label = above:{}] {};
    \node[place] (p*2) at (0,5) [label = above:{}] {};
    \node[place] (p*3) at (6,6.5) [label = above:{}] {};
    \node[place] (p*4) at (6,1.5) [label = above:{}] {};
    \node[place] (p1*) at (2,1.5) [label = above:{}] {};
    \node[place] (p2*) at (2,6.5) [label = above:{}] {};
    \node[place] (p3*) at (8,5) [label = above:{}] {};
    \node[place] (p4*) at (8,3) [label = above:{}] {};
    \node[place] (p13) at (4,4) [label = above:{}] {};
    \node[place] (p14) at (4,3) [label = above:{}] {};
    \node[place] (p23) at (4,5) [label = above:{}] {};
    \node[place,token = 1] (tokp*1) at (p*1) {};
    \node[place,token = 1] (tokp*1) at (p*2) {};
    \node[place,token = 1] (tokp*1) at (p*3) {};
    \node[place,token = 1] (tokp*1) at (p*4) {};
    \node[transition] (t1) at (2,3) {$\trT[1]$}
    edge[pre] node {} (p*1)
    edge[post] node {} (p1*)
    edge[post] node {} (p13)
    edge[post] node {} (p14);
    \node[transition] (t2) at (2,5) {$\trT[2]$}
    edge[pre] node {} (p*2)
    edge[post] node {} (p2*)
    edge[post] node {} (p23);
    \node[transition] (t3) at (6,5) {$\trT[3]$}
    edge[pre] node {} (p*3)
    edge[post] node {} (p3*)
    edge[pre] node {} (p13)
    edge[pre] node {} (p23);
    \node[transition] (t4) at (6,3) {$\trT[4]$}
    edge[pre] node {} (p*4)
    edge[post] node {} (p4*)
    edge[pre] node {} (p14);
  \end{tikzpicture}
  \caption{Occurrence net for the blockchain $\bcB = \txT[1] \txT[2] \txT[3] \txT[4]$ of the ERC-721 token.}
  \label{fig:erc721:petri}
\end{figure}

\section{Related work and conclusions}
\label{sec:conclusions}

We have proposed a static approach to improve the performance
of blockchains by concurrently executing transactions.
We have started by introducing a model of transactions and blockchains.
We have defined two transactions to be \emph{swappable}
when inverting their order does not affect the blockchain state.
We have then introduced a static approximation of swappability, 
based on a static analysis of the sets of keys read/written by transactions.
We have rendered concurrent executions of a sequence of transactions as
\emph{step firing sequences} in the associated occurrence net.
Our main technical result, \Cref{th:bc-to-pnet}, shows that
these concurrent executions are semantically equivalent to the
sequential one.

We can exploit our results in practice
to improve the performances of miners and validators. 
Miners should perform the following steps to mine a block:
\begin{enumerate} 
\item \label{item:txpar:miners:gather}
  gather from the network a set of transactions,
  and put them in an arbitrary linear order $\bcB$, which is the mined block;
\item \label{item:txpar:miners:pswap}
  compute the relation $\pswapWR$ on $\bcB$, using a 
  static analysis of read/written keys;
\item \label{item:txpar:miners:pnet}
  construct the occurrence net $\PNet{\pswapWR}{\bcB}$;
\item \label{item:txpar:miners:exec}
  execute transactions concurrently according to the occurrence net,
  exploiting the available parallelism.
\end{enumerate}

The behaviour of validators is almost identical to that of miners,
except that in step (1), rather than choosing the order of transactions,
they should adhere to the ordering of the mined block $\bcB$.
Note that in the last step, validators can execute
any maximal step firing sequence which is coherent
with their degree of parallelism:
item~\ref{th:bc-to-pnet:maximal} of \Cref{th:bc-to-pnet} ensures
that the resulting state is equal to the state obtained by the miner.
The experiments in~\cite{Dickerson17podc} suggest that parallelization
may lead to a significant improvement of the performance of nodes:
the benchmarks on a selection of representative contracts
show an overall speedup of 1.33x for miners and 1.69x for validators,
using only three cores.

Note that malevolent users could attempt a denial-of-service attack by 
publishing contracts which are hard to statically analyse,
and therefore are not suitable for parallelization. 
This kind of attacks can be mitigated by adopting a mining strategy that
gives higher priority to parallelizable transactions.

\paragraph{Applying our approach to Ethereum}

Applying our theory to Ethereum would require
a static analysis of read/written keys at the level of EVM bytecode.
As far as we know, the only tool implementing such an analysis
is ES-ETH~\cite{Marcia19eseth}.
However, the current version of the tool has several limitations, 
like \eg the compile-time approximation of dictionary keys and of 
values shorter than 32 bytes,
which make ES-ETH not directly usable to the purposes of our work.
In general, precise static analyses at the level of the 
Ethereum bytecode are difficult to achieve,
since the language has features like dynamic dispatching and pointer aliasing
which are notoriously a source of imprecision for static analysis.
However, coarser approximations of read/written keys
may be enough to speed-up the execution of transactions. 
For instance, in Ethereum, blocks typically contain many transactions 
which transfer tokens between participants,
and many of them involve distinct senders and receivers.
A relatively simple analysis of the code of token contracts
(which is usually similar to that in~\Cref{sec:erc721})
may be enough to detect that these transactions are swappable.

Aiming at minimality, our model does not include the \emph{gas mechanism},
which is used in Ethereum to pay miners for executing contracts.
The sender of a transaction deposits into it some crypto-currency,
to be paid to the miner which appends the transaction to the blockchain.
Each instruction executed by the miner consumes part of this deposit;
when the deposit reaches zero, the miner stops executing the transaction.
At this point, all the effects of the transaction
(except the payment to the miner) are rolled back.
Our transaction model could be easily extended with a gas mechanism,
by associating a cost to statements and recording the gas consumption in the environment.
Remarkably, adding gas does not invalidate approximations of read/written keys 
which are correct while neglecting gas.
However, a gas-aware analysis may be more precise of a gas-oblivious one:
for instance, in the statement
$\cmdIfTE{\valK}{\cmdCall{\funF[long]()}{}{}; \cmdAss{\valX}{1}}{\cmdAss{\valY}{1}}$
(where $\funF[long]$ is a function which exceeds the available gas)
a gas-aware analysis would be able to detect that x is not written.

\paragraph{Related work}

A few works study how to optimize the execution of smart contracts 
on Ethereum, using dynamic techniques adopted from software transactional 
memory~\cite{Anjana19pdp,Dickerson17podc,Dickerson18eatcs}.
These works are focussed on empirical aspects
(\eg, measuring the speedup obtained on a given benchmark), 
while we focus on the theoretical counterpart.
In \cite{Dickerson17podc,Dickerson18eatcs}, miners execute
a set of transactions speculatively in parallel,
using abstract locks and inverse logs to dynamically discover conflicts
and to recover from inconsistent states.
The obtained execution is guaranteed to be equivalent to a serial
execution of the same set of transactions.
The work \cite{Anjana19pdp} proposes a conceptually similar technique, but
based on optimistic software transactional memory.
Since speculative execution is non-deterministic,
in both approaches miners need to communicate the chosen schedule of transactions
to validators, to allow them to correctly validate the block.
This schedule must be embedded in the mined block:
since Ethereum does not support this kind of block metadata,
these approaches would require a ``soft-fork'' of the blockchain
to be implemented in practice.
Instead, our approach is compatible with the current Ethereum,
since miners only need to append transactions to the blockchain.
Compared to~\cite{Dickerson17podc,Anjana19pdp}, 
where conflicts are detected dynamically, our approach relies
on a static analysis to detect potential conflicts.
Since software transactional memory introduces a run-time overhead,
in principle a static technique could allow for faster executions,
at the price of a preprocessing phase.
Saraph and Herlihy~\cite{Saraph19arxiv} study the effectiveness of speculatively executing smart contracts
in Ethereum.
They sample past blocks of transactions (from July 2016 to December 2017), 
replay them by using a speculative execution engine, and measure the speedup obtained by parallel execution. 
Their results show that simple speculative strategies yield non-trivial speed-ups.
Further, they note that many of the data conflicts (\ie concurrent read/write accesses to the same state location)
arise in periods of high traffic, and they are caused by a small number of popular contracts, like \eg tokens.   

In the permissioned setting, 
Hyperledger Fabric \cite{Androulaki18eurosys} 
follows the ``execute first and then order'' paradigm:
transactions are executed speculatively, 
and then their ordering is checked for correctness~\cite{Fabric19rw}.
In this paradigm, appending a transaction requires a few steps.
First, a client proposes a transaction to a set of ``endorsing'' peers, 
which simulate the transaction without updating the blockchain.
The output of the simulation includes 
the state updates of the transaction execution,
and the sets of read/written keys.
These sets are then signed by the endorsing peers, and returned to the client,
which submits them to the ``ordering'' peers. 
These nodes order transactions in blocks, 
and send them to the ``committing'' peers, which validate them. 
A block $\txT[1] \cdots \txT[n]$ is valid when, 
if a key $\valK$ is read by transaction $\txT[i]$, then 
$\valK$ has not been written by a transaction $\txT[j]$ with $j<i$.
Finally, validated blocks are appended to the blockchain.
Our model is coherent with Ethereum, 
which does not support speculative execution.

\paragraph{Future works}

A relevant line of research is the design of 
domain-specific languages for smart contracts 
that are directly amenable to techniques that, like ours,
increase the degree of concurrency of executions.
For this purpose, the language should support static analyses of
read/written keys, like the one we use to define the strong swappability relation.
Although the literature describes various static analyses of smart contracts, 
most of them are focussed on finding security vulnerabilities,
rather than enhancing concurrency.

Outside the realm of smart contracts, 
a few papers propose static analyses of read/written variables.
The paper~\cite{Dias11} describes an analysis based on separation logic,
and applies it to resolve conflicts in the setting of
\emph{snapshot isolation} for transactional memory in Java.
When a conflict is detected, the read/write sets are used to determine how the
code can be modified to resolve it.
The paper~\cite{CheremCG08} presents a static analysis to infer read
and write locations in a C-like language with atomic sections.
The analysis is used to translate atomic sections into standard lock operations.
The design of new smart contract languages could take advantage of these analyses.


\paragraph{Acknowledgements} 
Massimo Bartoletti is partially supported 
by Aut.\ Reg.\ Sardinia projects \textit{``Smart collaborative engineering''} and \textit{``Sardcoin''}.
Letterio Galletta is partially supported by IMT Lucca project \emph{``PAI VeriOSS''} and by MIUR project PRIN 2017FTXR7S
\textit{``Methods and Tools for Trustworthy Smart Systems''}.
Maurizio Murgia is partially supported by MIUR PON \textit{``Distributed Ledgers for Secure Open Communities''}.

\bibliographystyle{splncs04}
\bibliography{main}

\iftoggle{arxiv}{
  \newpage
  \appendix
  \section{Proofs}\label{sec:proofs}

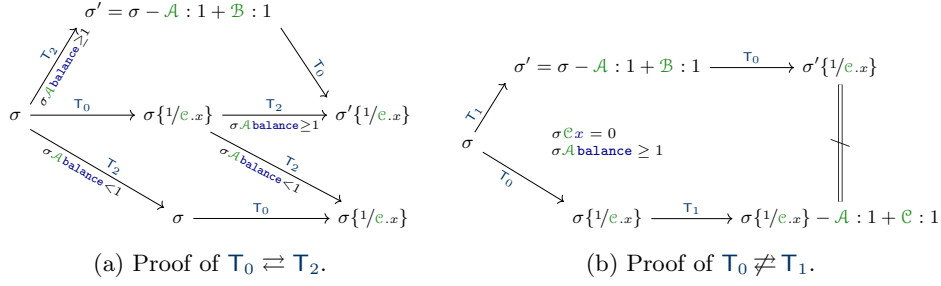
\begin{figure*}[t]
  \centering
  \hspace{-10pt}
  \begin{subfigure}[b]{0.3\textwidth}
    \adjustbox{scale=0.8,center}{
      \begin{tikzcd}[column sep=normal, row sep=huge]
        & \bcSti = \bcSt - \pmvA:1 + \pmvB:1  \ar[rd, "{\txT[0]}" sloped, end anchor=north west, start anchor=south east] & \\
        \bcSt
        \ar[r, "{\txT[0]}"]
        \ar[ur, start anchor=north east, end anchor=south west,"{\bcSt\pmvA \varBalance \geq 1}"' pos=0, sloped, "{\txT[2]}" near end]
        \ar[dr, start anchor=south east, end anchor=north west,"{\txT[2]}", sloped, "\bcSt\pmvA \varBalance < 1"' near end]
        & \bcSt\setenum{\bind{\cmvA.x}{1}} \ar[r, "\bcSt\pmvA \varBalance \geq 1"', sloped, "{\txT[2]}"] \ar[rd, "\bcSt\pmvA \varBalance < 1"' near end, sloped, "{\txT[2]}"]
        & \bcSti\setenum{\bind{\cmvA.x}{1}}  \\
        & \bcSt \ar[r, "{\txT[0]}"]
        & \bcSt\setenum{\bind{\cmvA.x}{1}}
      \end{tikzcd}
    }
    \subcaption{Proof of $\txT[0] \swap \txT[2]$.}
  \end{subfigure}
  \hspace{70pt}
  \begin{subfigure}[b]{0.33\textwidth}
    \adjustbox{scale=0.8,center}{   
      \begin{tikzcd}[column sep=small, row sep=normal, negated/.style={
          decoration={markings,
            mark= at position 0.5 with {
              \node[transform shape] (tempnode) {$\slash$};
            }
          },
          postaction={decorate}
        }
        ]
        & \bcSti = \bcSt - \pmvA:1 + \pmvB:1 \ar[r, "{\txT[0]}"] & \bcSti\setenum{\bind{\cmvA.x}{1}} \ar[dd, equal, negated] \\
        \bcSt
        \ar[ru, end anchor=south west, "{\txT[1]}", sloped]  \ar[rd, end anchor=north west, "{\txT[0]}"', sloped]
        &
        \begin{scriptsize}
          \begin{array}{l}
            \bcSt\cmvA \valX = 0 \\
            \bcSt\pmvA\varBalance \geq 1
          \end{array}
        \end{scriptsize}
        &
        \\
        &  \bcSt\setenum{\bind{\cmvA.x}{1}} \ar[r, "{\txT[1]}"] & \bcSt\setenum{\bind{\cmvA.x}{1}} - \pmvA:1 + \cmvA:1
      \end{tikzcd}
    }
    \subcaption{Proof of $\txT[0] \nswap \txT[1]$.}
  \end{subfigure}
  \caption{Proofs for~\Cref{ex:swap}. A transition $\txT$ from $\bcSt$ can be taken only if the guard below the arrow is satisfied in $\bcSt$.}
  \label{fig:swap:nocongruence}
\end{figure*}

\begin{figure}[t]
  \adjustbox{scale=0.9,center}{
    \begin{tabular}{c}
      \begin{tikzcd}[column sep=scriptsize, row sep=normal, negated/.style={
          decoration={markings,
            mark= at position 0.5 with {
              \node[transform shape] (tempnode) {$\slash$};
            }
          },
          postaction={decorate}}]
        & & & \bcStii - \pmvA:1 + \pmvB:1 \ar[dd, equal, negated] \\
        \bcSt
        \ar[r, bend left=45, "{\txT[2]}" sloped, end anchor=north west]
        \ar[r, bend right=45, "{\txT[1]}"' sloped, end anchor=south west]
        & \bcSti
        \ar[r, "{\txT[0]}"]
        & \bcStii
        \ar[ru, "{\txT[2]}" pos=0.7, start anchor=north east, end anchor=south west]
        \ar[rd, "{\txT[1]}" pos=0.3, start anchor=south east, end anchor=north west] \\
        & & & \bcStii - \pmvA:1 + \cmvA:1\\
        \small
        &
        &
      \end{tikzcd}
      \\[-10pt]
      $\bcSti = \bcSt - \pmvA:1 + \pmvB:1$
      \;\;
      $\bcStii = \bcSti\setenum{\bind{\cmvA.x}{1}}$
    \end{tabular}
  }
  \caption{Proof of $\txT[1] \txT[0] \txT[2] \not\sim \txT[2] \txT[0] \txT[1]$ for \Cref{ex:swap:mazurkiewicz}.}
  \label{fig:swap:bc-noncongruence}
\end{figure}

\lemsimequiv*
\begin{proof}
  Items~\ref{lem:sim:equiv:equiv} and~\ref{lem:sim:equiv:subseteq} are trivial.
  The inclusion $\sim_{\QmvU} \,\subseteq\, \sim$ is trivial,
  and $\sim\, \subseteq\, \sim_{\QmvU}$,
  follows from item~\ref{lem:sim:equiv:subseteq}.
  \qed
\end{proof}

\lemsimappend*
\begin{proof}
  Direct from the fact that semantics of statements is a function,
  and it only depends on the blockchain states after the execution of
  $\bcB$ and $\bcBi$, which are equal in any $\bcSt$ since $\bcB \sim \bcBi$.
  \qed
\end{proof}

\thswapmazurkiewicz*
\begin{proof}
  By definition, $\equivStSeq[\swap]$ is the least equivalence 
  relation closed under the rules:
  \[
  \begin{array}{c}
    \irule{}{\bcEmpty \equivStSeq[\swap] \bcEmpty} 
    \nrule{[$\equivStSeq$0]}
    \quad
    \irule{}{\txT \equivStSeq[\swap] \txT} 
    \nrule{[$\equivStSeq$1]}
    \quad
    \irule{\txT \swap \txTi}{\txT \txTi \equivStSeq[\swap] \txTi \txT} 
    \nrule{[$\equivStSeq$2]}
    \quad
    \irule{\bcB[0] \equivStSeq[\swap] \bcBi[0] \quad 
    \bcB[1] \equivStSeq[\swap] \bcBi[1]}
    {\bcB[0] \bcB[1] \equivStSeq[\swap] \bcBi[0] \bcBi[1]} 
    \nrule{[$\equivStSeq$3]}
  \end{array}
  \] 
  Let $\bcB \equivStSeq[\swap] \bcBi$. 
  We have to show $\bcB \sim \bcBi$.
  We proceed by induction on the rules above. 
  The case for rules
  \nrule{[$\equivStSeq$0]} and \nrule{[$\equivStSeq$1]} follows by the fact
  that $\sim$ is an equivalence relation (\Cref{lem:sim:equiv}) and hence reflexive.
  The case for rule \nrule{[$\equivStSeq$2]} follows immediately by \Cref{def:swap}.
  For rule \nrule{[$\equivStSeq$3]}, first note that $\bcB = \bcB[0] \bcB[1]$ and
  $\bcBi = \bcBi[0] \bcBi[1]$. By the induction hypothesis it follows that:
  \[\bcB[0] \sim \bcBi[0] \quad \text{and} \quad \bcB[1] \sim \bcBi[1]\]
  Therefore, by two applications of \Cref{lem:sim:append}:
  \[
  \bcB = \bcB[0] \bcB[1] \sim \bcB[0] \bcBi[1] \sim \bcBi[0] \bcBi[1] = \bcBi
  \tag*{\qed}
  \]
\end{proof}

\begin{lemma}
  \label{lem:safeapprox}
  Let $\bullet \in \setenum{r,w}$. Then:
  \begin{enumerate}[(a)]
  \item \label{lem:safeapprox:subseteq}
    if $\safeapprox[\bullet]{\QmvA}{\txT}$ and $\QmvA \subseteq \QmvAi$,
    then $\safeapprox[\bullet]{\QmvAi}{\txT}$;
  \item \label{lem:safeapprox:cap}
    if $\safeapprox[w]{\QmvA}{\txT}$
    and $\safeapprox[w]{\QmvB}{\txT}$,
    then $\safeapprox[w]{\QmvA \cap \QmvB}{\txT}$.
  \end{enumerate}
\end{lemma}
\begin{proof}
  \Cref{lem:safeapprox:subseteq}.
  For the case $\bullet = w$, let $\wapprox{\QmvA}{\txT}$ and 
  $\QmvA \subseteq \QmvAi$. 
  Let $\QmvB$ be such that $\QmvB \cap \QmvAi = \emptyset$.
  We have to show that $\txT \sim_{\QmvB} \bcEmpty$.
  Since $\QmvA \subseteq \QmvAi$,
  it must be $\QmvB \cap \QmvA = \emptyset$. Then, since $\wapprox{\QmvA}{\txT}$,
  it must be $\txT \sim_{\QmvB} \bcEmpty$, as required.
  For the case $\bullet = r$, let $\rapprox{\QmvA}{\txT}$ and 
  $\QmvA \subseteq \QmvAi$. 
  We have to show that, for all $\bcB[1],\bcB[2]$,
  if $\bcB[1] \sim_{\QmvAi} \bcB[2]$ and $\bcB[1] \sim_{\QmvB} \bcB[2]$,
  then $\bcB[1]\txT \sim_{\QmvB} \bcB[2]\txT$. But this follows immediately
  by the fact that $\QmvA \subseteq \QmvAi$ and $\rapprox{\QmvA}{\txT}$.
  
  \Cref{lem:safeapprox:cap}. 
  Let $\QmvC$ be such that $\QmvC \cap (\QmvA \cap \QmvB) = \emptyset$.
  Since $\safeapprox[w]{\QmvA}{\txT}$ and $(\QmvC \setminus \QmvA) \cap \QmvA =
  \emptyset$, it must be:
  \[
  \txT \sim_{\QmvC \setminus \QmvA} \bcEmpty
  \]
  Similarly, since $\safeapprox[w]{\QmvB}{\txT}$ and 
  $(\QmvC \setminus \QmvB) \cap \QmvB = \emptyset$, we have that:
  \[
  \txT \sim_{\QmvC \setminus \QmvB} \bcEmpty
  \]
  By assumption $\QmvC \cap (\QmvA \cap \QmvB) = \emptyset$,
  then $(\QmvC \setminus \QmvA) \cup (\QmvC \setminus \QmvB) = \QmvC$.
  By~\Cref{def:sim}, we conclude that:
  \[
  \txT \sim_{\QmvC} \txT \sim_{(\QmvC \setminus \QmvA) \cup (\QmvC \setminus \QmvB)} \bcEmpty
  \tag*{\qed}
  \]
\end{proof}

The following~\namecref{ex:safeapprox:read-not-cap}
shows that part~\ref{lem:safeapprox:cap} of~\Cref{lem:safeapprox}
does not hold for read approximations.

\begin{example}
  \label{ex:safeapprox:read-not-cap}
  Let $\cmvA$ be a contract with exactly two procedures:
  \begin{align*}
    \contrFunSig{\funF}{\varX}
    & \; \{ \cmdAss{\valK}{x}; \cmdAss{\valKi}{x} \}
    \\
    \contrFunSig{\funG}{}
    & \; \{\cmdIfTE{\expLookup{\valK} \neq \pmvA}{\cmdSend{\expLookup{\varBalance}}{\pmvB}}{\cmdSkip} \}
  \end{align*}
  and let $\txT = \ethtx{0}{\pmvA}{\cmvA}{\funG}{}$.
  Note that, in any reachable state $\bcSt$,
  it must be $\bcSt \cmvA \valK = \bcSt \cmvA \valKi$.
  Let $\QmvB$ be such that $\bcB \sim_{\QmvB} \bcBi$,
  and let $\bcSt = \semBc{\bcB}{}$, $\bcSti = \semBc{\bcBi}{}$,
  $\valN = \bcSt \cmvA \varBalance$, and
  $\valNi = \bcSti \cmvA \varBalance$.
  Appending $\txT$ to $\bcB$ and $\bcBi$ will result in:
  \begin{align*}
    \semTx{\txT}{\bcSt} & = \begin{cases}
      \bcSt - \cmvA:\valN + \pmvB:\valN & \text{if $\bcSt \cmvA \valK \neq \pmvA$} \\
      \bcSt & \text{otherwise}
    \end{cases}
    \\
    \semTx{\txT}{\bcSti} & = \begin{cases}
      \bcSti - \cmvA:\valNi + \pmvB:\valNi & \text{if $\bcSti \cmvA \valK \neq \pmvA$} \\
      \bcSti & \text{otherwise}
    \end{cases}
  \end{align*}
  If $\bcB \sim_{\setenum{\valK}} \bcBi$, then the conditions
  $\bcSt \cmvA \valK \neq \pmvA$ and $\bcSti \cmvA \valK \neq \pmvA$
  are equivalent.
  Therefore,
  $\sem{\bcB \txT} = \semTx{\txT}{\bcSt} \sim_{\QmvB} \semTx{\txT}{\bcSti} = \sem{\bcBi \txT}$,
  and so we have proved that $\rapprox{\setenum{\valK}}{\txT}$.
  Similarly, we obtain that $\rapprox{\setenum{\valKi}}{\txT}$,
  since $\valK$ and $\valKi$ are always bound to the same value.
  Note however that $\setenum{\valK} \cap \setenum{\valKi} = \emptyset$
  is \emph{not} a safe approximation of the keys read by $\txT$.
  For instance, if $\bcSt \cmvA \valK = \pmvA \neq \bcSti \cmvA \valKi$
  and $\bcSt \cmvA \varBalance = \bcSti \cmvA \varBalance$,
  then appending $\txT$ to $\bcB$ or to $\bcBi$ results in states
  which differ in the balance of $\cmvA$.
  \qed
\end{example}

\begin{lemma}
  \label{lem:dummy-txTi}
  Let $\wapprox{\QmvA[\txT]^w}{\txT}$, $\wapprox{\QmvA[\txTi]^w}{\txTi}$,
  and $\rapprox{\QmvA[\txT]^r}{\txT}$. 
  Then:
  \[
  \QmvA[\txTi]^w \cap \QmvA[\txT]^r = \emptyset = \QmvA[\txT]^w \cap S_{\txT'}^w
  \quad \implies \quad
  \txT \sim_{\QmvA[\txT]^w} \txTi\txT
  \]
\end{lemma}
\begin{proof}
  Since $\wapprox{\QmvA[\txTi]^w}{\txTi}$
  and $\QmvA[\txTi]^w \cap \QmvA[\txT]^w = \emptyset$,
  by \Cref{def:safeapprox} we have
  $\txTi \sim_{\QmvA[\txT]^w} \bcEmpty$.
  We prove that $\txTi \sim_{\QmvA[\txT]^r} \bcEmpty$.
  By contradiction, assume that
  $\txTi \not\sim_{\qmvA} \bcEmpty$ for some $\qmvA \in \QmvA[\txT]^r$.
  Since $\wapprox{\QmvA[\txTi]^w}{\txTi}$, it must be
  $\qmvA \in \QmvA[\txTi]^w$ ---
  contradiction, since $\QmvA[\txTi]^w \cap \QmvA[\txT]^r = \emptyset$.
  Since $\rapprox{\QmvA[\txT]^r}{\txT}$, 
  by \Cref{def:safeapprox} we conclude that
  $\txTi \txT \sim_{\QmvA[\txT]^w} \txT$.
  \qed
\end{proof}

\begin{lemma}
  \label{lem:commutativity-writes}
  \(
  \txT[1] \pswapWR \txT[2] \implies \txT[1]\txT[2] \sim_{\wset{\txT[1]}} \txT[2]\txT[1]
  \)
\end{lemma}
\begin{proof}
  By \Cref{def:pswapWR}, it must be: 
  \[\wset{\txT[1]} \cap \wset{\txT[2]} = \emptyset\]
  So, since $\wapprox{\wset{\txT[2]}}{\txT[2]}$, by \Cref{def:safeapprox} we have:
  \[
  \bcEmpty \sim_{\wset{\txT[1]}} \txT[2] 
  \]
  and so, since $\sim_{\wset{\txT[1]}}$ is a congruence:
  \[
  \txT[1] \sim_{\wset{\txT[1]}} \txT[1]\txT[2] 
  \]
  By \Cref{def:pswapWR}, it must be: 
  \[
  \wset{\txT[2]} \cap \rset{\txT[1]} = \emptyset = 
  \wset{\txT[1]} \cap \wset{\txT[2]}
  \] 
  Then, by \Cref{lem:dummy-txTi}:
  \[
  \txT[1] \sim_{\wset{\txT[1]}} \txT[2]\txT[1]
  \]
  By simmetry and transitivity of $\sim$ (\Cref{lem:sim:equiv}), we conclude:
  \[
  \txT[1]\txT[2] \sim_{\wset{\txT[1]}} \txT[2]\txT[1]
  \tag*{\qed}
  \]
\end{proof}

\begin{lemma}
  \label{th:pswapWR-implies-swap}
  $\txT[1] \pswapWR \txT[2] \implies \txT[1] \swap \txT[2]$
\end{lemma}
\begin{proof}
  By applying \Cref{lem:commutativity-writes} twice:
  \[
  \txT[1]\txT[2] \sim_{\wset{\txT[1]}} \txT[2]\txT[1] \qquad
  \txT[2]\txT[1] \sim_{\wset{\txT[2]}} \txT[1]\txT[2]
  \]
  Let $\QmvA = \QmvU \setminus (\wset{\txT[1]} \cup \wset{\txT[2]})$.
  Since $\QmvA \cap \wset{\txT[1]} = \emptyset = \QmvA \cap \wset{\txT[2]}$,
  by applying~\Cref{def:safeapprox} twice:
  \[
  \bcEmpty \sim_{\QmvA} \txT[1]
  \qquad
  \bcEmpty \sim_{\QmvA} \txT[2]
  \]
  Then, since $\sim_{\QmvA}$ is a congruence:
  \[
  \txT[1]\txT[2] \sim_{\QmvA} \txT[2]\txT[1]
  \]
  Summing up:
  \[
  \txT[1]\txT[2] \sim_{\QmvA \cup (\wset{\txT[1]} \cup \wset{\txT[2]})} \txT[2]\txT[1]
  \]
  from which we obtain the thesis, since $\sim_{\QmvU} = \sim$.
  \qed
\end{proof}

\begin{lemma}
  \label{lem:pswap-implies-pswapWR}
  \(
  \txT[1] \pswap{}{} \txT[2]
  \implies
  \exists \wset{}, \rset{} 
  \, : \,
  \txT[1] \pswapWR \txT[2]
  \)
\end{lemma}
\begin{proof}
  Straightforward from~\Cref{def:pswap} and~\Cref{def:pswapWR}.
  \qed
\end{proof}

\thpswapimpliesswap*
\begin{proof}
  By~\Cref{lem:pswap-implies-pswapWR} and \Cref{th:pswapWR-implies-swap}.
  \qed
\end{proof}

\lempswapWRimpliespswap*
\begin{proof}
  Trivial.
  \qed
\end{proof}

\thpswapmazurkiewicz*
\begin{proof}
  Straightforward by~\Cref{th:pswap-implies-swap}.
  \qed
\end{proof}

\begin{lemma}
  \label{lem:equivRel-implies-equiv}
  \(
  \bcB \equivStSeq[\pswap{}{}] \bcBi
  \;\;\implies\;\;
  \bcB \sim \bcBi
  \)
\end{lemma}
\begin{proof}
  Direct by \Cref{th:swap:mazurkiewicz,th:pswap:mazurkiewicz}.
  \qed
\end{proof}

\begin{lemma}
  \label{lem:merge-comm-monoid}
  $\mrg$ is commutative and associative, 
  with $\lambda \qmvA. \bot$ as neutral element.
\end{lemma}
\begin{proof}
  Trivial.
  \qed
\end{proof}

\begin{lemma}
  \label{lem:mrg-seq}
  If $\mSubst[1] \mrg \mSubst[2] = \mSubst$, then 
  $\mSubst = \mSubst[1]\mSubst[2]$.
\end{lemma}
\begin{proof}
  Since $\mSubst[1] \mrg \mSubst[2]$ is defined, it must be
  $\keys{\mSubst[1]} \cap \keys{\mSubst[2]} = \emptyset$.
  Let $\qmvA$ be a qualified key. We have two cases:
  \begin{itemize}
  \item $\qmvA \in \keys{\mSubst}$. Since $\keys{\mSubst} = \keys{\mSubst[1]} 
    \cup \keys{\mSubst[2]}$, we have two subcases:
    \begin{itemize}
    \item $\qmvA \in \keys{\mSubst[1]}$. Then, $\mSubst \qmvA = \mSubst[1]\qmvA$.
      By disjointness, $\qmvA \not\in \keys{\mSubst[2]}$, and hence
      $\mSubst[1]\mSubst[2]\qmvA = \mSubst[1]\qmvA$.
    \item $\qmvA \in \keys{\mSubst[2]}$. Then, $\mSubst \qmvA = \mSubst[2]\qmvA = 
      \mSubst[1]\mSubst[2]\qmvA$.
    \end{itemize}
  \item $\qmvA \not\in \keys{\mSubst}$. 
    Then, $\qmvA \not\in \keys{\mSubst[1]}$,
    $\qmvA \not\in \keys{\mSubst[2]}$, 
    and so $\mSubst \qmvA = \bot = \mSubst[1]\mSubst[2]\qmvA$.
    \qed
  \end{itemize}
\end{proof}

\begin{lemma}
  If $\bcB[1] \seqn \txsA[1]$ and $\bcB[2] \seqn \txsA[2]$, then 
  $\bcB[1]\bcB[2] \seqn (\txsA[1] \cup \txsA[2])$.
\end{lemma}
\begin{proof}
  By induction on $\card{\bcB[2]}$. For the base case, it must be 
  $\bcB[2] = \bcEmpty$ and hence $\txsA[2] = \emptyset$. Then,
  $\bcB[1]\bcB[2] = \bcB[1]$ and $\txsA[1] \cup \txsA[2] = \txsA[1]$.
  Therefore, the thesis coincides with the first hypothesis.
  For the induction case, it must be $\bcB[2] = \bcBi[2]\txT$, with
  $\card{\bcBi[2]} = n$. Furthermore, it must be
  $\txsA[2] = \setenum{\txT} \cup \txsAi[2]$, for some $\txsAi[2]$ such that
  $\bcBi[2] \seqn \txsAi[2]$. By the induction hypothesis:
  \[
  \bcB[1]\bcBi[2] \seqn (\txsA[1] \cup \txsAi[2])
  \]
  Then:
  \[
  \bcB[1]\bcBi[2]\txT = \bcB[1]\bcB[2] \seqn 
  (\setenum{\txT} \cup \txsA[1] \cup \txsAi[2]) = \txsA[1] \cup \txsA[2]
  \tag*{\qed}
  \]
\end{proof}


\begin{lemma}
  \label{lem:strong-swap-chain-W-R}
  Let $\bcB$ and $\txT$ be such that 
  $\bcB = \bcB[1]\txTi\bcB[2] \implies \txT \pswapWR \txTi$. Then, for all
  $\bcBi$, $\semBc{\bcBi}{} \sim_{\rset{\txT}}\semBc{\bcB}{\semBc{\bcBi}{}}$
  and $\semBc{\bcBi}{} \sim_{\wset{\txT}}\semBc{\bcB}{\semBc{\bcBi}{}}$.
\end{lemma}
\begin{proof}
  A simple induction on $\card{\bcB}$, using \Cref{def:safeapprox} 
  for the induction case.
  \qed
\end{proof}

\begin{lemma}
  \label{lem:strong-swap-seq}
  If $\txT \pswapWR \txTi$ for all $\txTi \in \txsA$ and $\bcB[1] \seqn \txsA$
  then, $\bcB = \bcB[1]\txTi\bcB[2] \implies \txT \pswapWR \txTi$.
\end{lemma}
\begin{proof}
  A simple induction on $\card{\txsA}$.
  \qed
\end{proof}

\begin{lemma}
  \label{lem:pi-swap}
  If $\rapprox{\qmvA}{\txT}$, $\bcB[1] \sim_{\qmvA} \bcB[2]$ and 
  $\bcB[1] \sim_{\qmvB} \bcB[2]$, then $\semBc{\bcB[1]\txT}{} \sim_{\qmvB} 
  \semBc{\bcB[1]}{}\WR{\semBc{\bcB[2]}{},\txT}$.
\end{lemma}
\begin{proof}
  Let $\mSubst[1] = \WR{\semBc{\bcB[1]}{},\txT}$ and 
  $\mSubst[2] = \WR{\semBc{\bcB[2]}{},\txT}$. By \Cref{def:wr}, 
  $\semBc{\bcB[1]\txT}{} = \semBc{\bcB[1]}{}\mSubst[1]$. Let $\qmvA \in \QmvB$.
  We have two cases:
  \begin{itemize}
  \item $\qmvA \in \keys{\mSubst[2]}$. 
    \[
    \semBc{\bcB[1]}{}\mSubst[2]\qmvA = \mSubst[2]\qmvA = 
    \semBc{\bcB[2]}{}\mSubst[2]\qmvA = \semBc{\bcB[2]\txT}{}\qmvA = 
    \semBc{\bcB[1]\txT}{}\qmvA
    \]
  \item $\qmvA \not\in \keys{\mSubst[2]}$.
    \[
    \semBc{\bcB[1]}{}\mSubst[2]\qmvA = \semBc{\bcB[1]}{}\qmvA = 
    \semBc{\bcB[2]}{}\qmvA = \semBc{\bcB[2]}{}\mSubst[2]\qmvA = 
    \semBc{\bcB[1]\txT}{}\qmvA
    \tag*{\qed}
    \]
  \end{itemize}
\end{proof}

\begin{definition}
  Let $\WR{}$ be a state updater, and
  let $\wset{}$ be such that $\forall \txT: \wapprox{\wset{\txT}}{\txT}$.
  We say that $\WR{}$ and $\wset{}$ are \emph{compatible} when
  $\forall \bcSt,\txT : \keys{\WR{\bcSt,\txT}} \subseteq \wset{\txT}$.
\end{definition}

We extend the semantics of transactions
to finite \emph{multisets} of transactions.
Hereafter, we denote with $\emptymset$ the empty multiset,
with $\msetenum{\txT[1],\ldots,\txT[n]}$ the multiset containing
$\txT[1],\ldots,\txT[n]$, and with
$A + B$ the sum between multisets,
\ie $(A+B)(x) = A(x) + B(x)$ for all $x$.

\begin{definition}[\textbf{Semantics of multisets of transactions}]
  \label{def:multiset-sem}
  We denote the semantics of a multiset of transactions $\txsA$,
  in a state $\bcSt$ and an update collector $\WR{}$,
  as $\msem{\WR{}}{\bcSt}{\txsA}$,
  where the partial function
  $\msem{\WR{}}{\bcSt}{\cdot}$ is defined as:
  \( \;
  \msem{\WR{}}{\bcSt}{\txsA}
  \, = \,
  \bcSt \bigoplus_{\txT \in \txsA} \WR{\bcSt,\txT}
  \).
  %
\end{definition}


Hereafter, we say that a multiset $\txsA$ is strongly swappable \wrt $\pswapWR$
if $\forall \txT \in \txsA, \forall \txTi \in \txsA \setminus \msetenum{\txT} : \txT \pswapWR \txTi$.

\begin{lemma}
  \label{lem:step-seq-union-aux}
  If $\txsA$ is strongly swappable \wrt $\pswapWR$, $\bcB \seqn \txsA$ and $\WR{}$
  is compatible with $\wset{}$ then,
  for all $\bcB[0]$: $\msem{\WR{}}{\semBc{\bcB[0]}{}}{\txsA} = 
  \semBc{\bcB}{\semBc{\bcB[0]}{}}$.
\end{lemma}
\begin{proof}
  By induction on $\card{\bcB}$. For the base case, it must be 
  $\bcB = \bcEmpty$ and $\txsA = \emptyset$, and hence 
  $\msem{\WR{}}{\semBc{\bcB[0]}{}}{\emptyset} = \semBc{\bcB[0]}{} = 
  \semBc{\bcEmpty}{\semBc{\bcB[0]}{}}$.
  For the induction case, it must be $\bcB = \bcBi\txT$, with 
  $\card{\bcBi} = n$. Clearly, 
  $\txsA = \setenum{\txT} \cup \txsAi$ for some $\txsAi$ such that
  $\bcBi \seqn \txsAi$. Let $\WR{\semBc{\bcB[0]}{},\txT} = \mSubst[\txT]$.
  By the induction hypothesis:
  \begin{equation}
    \label{eq:step-seq-union-3}
    \msem{\WR{}}{\semBc{\bcB[0]}{}}{\txsAi} = \semBc{\bcBi}{\semBc{\bcB[0]}{}}
  \end{equation}
  Note that:
  \begin{equation}
    \label{eq:step-seq-union-2}
    \msem{\WR{}}{\semBc{\bcB[0]}{}}{\txsAi} = \semBc{\bcB[0]}{}\mSubsti
  \end{equation} 
  Where $\mSubsti \bigoplus_{\txTi \in \txsAi} \WR{\semBc{\bcB[0]}{},\txTi})$.
  Let $\WR{\semBc{\bcB[0]}{},\txT} = \mSubst[\txT]$.
  Since $\txsA$ is strongly swappable \wrt $\pswapWR$ and $\WR{}$
  is compatible with $\wset{}$, it must be $\keys{\mSubsti}\cap 
  \keys{\mSubst[\txT]} = \emptyset$, and hence $(\mSubsti\mrg \mSubst[\txT])$ 
  is defined. 
  Then, it must be:
  \begin{align}
    \nonumber
    \msem{\WR{}}{\semBc{\bcB[0]}{}}{\txsA} 
    & = \semBc{\bcB[0]}{}(\mSubsti\mrg \mSubst[\txT]) 
    \\
    \nonumber
    & = \semBc{\bcB[0]}{}\mSubsti \mSubst[\txT] 
    && \text{By \Cref{lem:mrg-seq}}
    \\
    \nonumber
    & = \msem{\WR{}}{\semBc{\bcB[0]}{}}{\txsAi}\mSubst[\txT] 
    && \text{By \Cref{eq:step-seq-union-2}}
    \\
    \label{eq:step-seq-union-1}
    & =  \semBc{\bcBi}{\semBc{\bcB[0]}{}}\mSubst[\txT] 
    && \text{By \Cref{eq:step-seq-union-3}}
  \end{align}
  We have that:
  \[
  \semBc{\bcB}{\semBc{\bcB[0]}{}} = \semBc{\bcBi\txT}{\semBc{\bcB[0]}{}} =
  \semBc{\bcBi}{\semBc{\bcB[0]}{}}\mSubsti[\txT]
  \]
  Where $\mSubsti[\txT] = \WR{\semBc{\bcBi[0]}{\semBc{\bcB[0]}{}},\txT}$.
  As $\keys{\mSubst[\txT]} \subseteq \wset{\txT}$ and 
  $\keys{\mSubst[\txT]} \subseteq \wset{\txT}$, it follows
  immediately that $\semBc{\bcBi}{\semBc{\bcB[0]}{}}\mSubst[\txT] \sim_{\qmvA}
  \semBc{\bcB}{\semBc{\bcB[0]}{}}$ for all $\qmvA \not\in \wset{\txT}$.
  It remains to show that $\semBc{\bcBi}{\semBc{\bcB[0]}{}}\mSubst[\txT] 
  \sim_{\wset{\txT}} \semBc{\bcB}{\semBc{\bcB[0]}{}}$.
  First note that, by \Cref{lem:strong-swap-seq,lem:strong-swap-chain-W-R}:
  \[
  \semBc{\bcB[0]}{} \sim_{\rset{\txT}} \semBc{\bcBi}{\semBc{\bcB[0]}{}} \qquad
  \semBc{\bcB[0]}{}\sim_{\wset{\txT}} \semBc{\bcBi}{\semBc{\bcB[0]}{}}
  \]
  Then, by \Cref{lem:pi-swap}:
  \[
  \semBc{\bcBi}{\semBc{\bcB[0]}{}}\mSubst[\txT] 
  \sim_{\wset{\txT}} \semBc{\bcB}{\semBc{\bcB[0]}{}}
  \]
  And hence:
  \begin{equation}
    \label{eq:step-seq-union-4}
    \semBc{\bcBi}{\semBc{\bcB[0]}{}}\mSubst[\txT] = \semBc{\bcB}{\semBc{\bcB[0]}{}}
  \end{equation}
  The thesis $\msem{\WR{}}{\semBc{\bcB[0]}{}}{\txsA} = 
  \semBc{\bcB}{\semBc{\bcB[0]}{}}$ follows by 
  \Cref{eq:step-seq-union-4,eq:step-seq-union-1}.
  \qed
\end{proof}

\begin{lemma}
  \label{lem:W-compatible-WRmin}
  Every $\wset{}$ is compatible with $\WRmin{}{}$.
\end{lemma}
\begin{proof}
  Trivial.
  \qed
\end{proof}

We now formalize when a blockchain $\bcB$ is a
serialization of a multiset of transactions $\txsA$.

\begin{definition}[\textbf{Serialization of multisets of transactions}]
  \label{def:mset-serial}
  We define the relation $\seqn$ between blockchains and multisets 
  of transactions as follows:
  \[
  \irule{}
  {\bcEmpty \seqn \emptymset}
  \qquad\qquad
  \irule{\bcB \seqn \txsA}
  {\bcB\txT \seqn (\msetenum{\txT} + \txsA)}
  \]
\end{definition}

The following~\namecref{th:step-seq-union} 
ensures that the parallel execution of strongly swappable transactions 
is equivalent to any sequential execution of them.
Hereafter, we say that a multiset $\txsA$ is \emph{strongly swappable}
if $\forall \txT \in \txsA, \forall \txTi \in \txsA - \msetenum{\txT} : \txT \pswap{}{} \txTi$.

\begin{restatable}{theorem}{thstepsequnion}
  \label{th:step-seq-union}
  If $\txsA$ is strongly swappable and
  $\bcB \seqn \txsA$, then,
  for all reachable $\bcSt$: $\msem{\WRmin{}}{\bcSt}{\txsA} =
  \semBc{\bcB}{\bcSt}$.
\end{restatable}
\begin{proof}
  Direct by \Cref{lem:step-seq-union-aux,lem:W-compatible-WRmin}.
  \qed
\end{proof}

A \emph{parellelized blockchain} $\bcsA$ 
is a finite sequence of multisets of transactions;
we denote with $\bcEmpty$ the empty sequence.
We extend the semantics of multisets (\Cref{def:multiset-sem}) to 
parallelized blockchains as follows.

\begin{definition}[\textbf{Semantics of parallelized blockchains}]
  \label{def:txpar:pbc-semantics}
  The semantics of parallelized blockchains is defined as follows:
  \[
  \msem{\WR{}}{\bcSt}{\bcEmpty} = \bcSt
  \qquad
  \msem{\WR{}}{\bcSt}{\txsA \bcsA} =
  \msem{\WR{}}{\scriptsize \msem{\WR{}}{\bcSt}{\txsA}}{\bcsA}
  \]
  We write $\msem{\WR{}}{}{\bcsA}$ for $\msem{\WR{}}{\bcStInit}{\bcsA}$,
  where $\bcStInit$ is the initial state.
\end{definition}

We also extend the serialization relation $\seqn$ (\Cref{def:mset-serial}) to
parallelized blockchains.

\begin{definition}[\textbf{Serialization of parallelized blockchains}]
  We define the relation $\seqn$ between blockchains and 
  parallelized blockchains as follows:
  \[
  \irule{}{\bcEmpty \seqn \bcEmpty} 
  \qquad\qquad
  \irule
  {\bcB[1] \seqn \txsA \quad \bcB[2] \seqn \bcsA}
  {\bcB[1]\bcB[2] \seqn \txsA\bcsA}
  \]
\end{definition}

The following~\namecref{th:seq-union}
states that our technique to parallelize the transactions in a blockchain
preserves its semantics.

\begin{restatable}{theorem}{thsequnion}
  \label{th:seq-union}
  If each multiset in $\bcsA$ is strongly swappable
  and $\bcB \seqn \bcsA$, then,
  for all reachable $\bcSt$: $\msem{\WRmin{}}{\bcSt}{\bcsA} =
  \semBc{\bcB}{\bcSt}$.
\end{restatable}
\begin{proof}
  By induction on the rule used for deriving $\bcB \seqn \bcsA$.
  \begin{itemize}
  \item Rule: $\irule{}{\bcEmpty \seqn \bcEmpty}$.
    \\[5pt]
    The thesis follows trivially, since:
    \[
    \semBc{\bcEmpty}{\bcSt}= \bcSt = \msem{\WRmin{}}{\bcSt}{\bcEmpty} 
    \]
  \item Rule: $\irule{\bcB[1] \seqn \txsA \quad \bcB[2] \seqn \bcsA}
    {\bcB[1]\bcB[2] \seqn \txsA\bcsA}$.
    \\[5pt] 
    By \Cref{th:step-seq-union}, for some reachable $\bcSti$:
    \[
    \semBc{\bcB[1]}{\bcSt}= \bcSti = \msem{\WRmin{}}{\bcSt}{\txsA} 
    \]
    By the induction hypothesis:
    \[
    \semBc{\bcB[2]}{\bcSti} = \msem{\WRmin{}}{\bcSti}{\bcsA} 
    \]
    The thesis then follows by:
    \[
    \semBc{\bcB[2]}{\bcSti} = \semBc{\bcB[1]\bcB[2]}{\bcSt} 
    \qquad
    \msem{\WRmin{}}{\bcSti}{\bcsA} = \msem{\WRmin{}}{\bcSt}{\txsA\bcsA}
    \tag*{\qed}
    \]
  \end{itemize}
\end{proof}

\begin{lemma}
  \label{lem:le-star-po}
  Let 
  $\PNet{\relR}{\txT[1] \cdots \txT[n]} = (\Places, \Transitions, \Arcs, \markM[0])$.
  Then $(\Transitions,<^*)$ is a partial order.
\end{lemma}
\begin{proof}
  Transitivity and reflexivity hold by definition. 
  For antisymmetricity, assume that 
  $(\txT[i],i) <^* (\txT[j],j)$ and $(\txT[j],j) <^* (\txT[i],i)$.
  Then, it is easy to verify that $i \leq j$ and $j \leq i$, and so $i = j$.
  Since $\txT[i]$ and $\txT[j]$ are uniquely determined by $i$ and $j$, we have
  that $\txT[i] = \txT[j]$. 
  Therefore, $(\txT[i],i) = (\txT[j],j)$, as required.
  \qed
\end{proof}

\lembctopnetonet*
\begin{proof}
  By~\Cref{def:bc-to-pnet}, the first three conditions
  of the definition of occurrence net are easy to verify.
  To prove that $\Arcs^*$ is acyclic, we proceed by contradiction. 
  Assume that there is a sequence 
  $\vec x = x_0,x_1,\hdots x_m$ such that $(x_{i},x_{i+1}) \in \Arcs$ 
  for all $0 \leq i < m$, and $x_0 = x_m$ with $m > 0$. 
  Note that the above sequence alternates between 
  transitions and places, and so, since $m > 0$, 
  at least one place and one transition occur in $\vec x$. 
  Further, a place between two transitions 
  $\trT \neq \trTi$ can exist only if $\trT < \trTi$. 
  Therefore, if $\trT,\trTi$ occur in $\vec x$, 
  it must be $\trT <^* \trTi$ and $\trTi <^* \trT$.
  So, if $\vec x$ contains at least two transitions, by \Cref{lem:le-star-po}, 
  we have a contradiction. If only one transition $\trT = (\txT,i)$ occurs in $\vec x$, 
  then there is a place of the form $(\trT,\trT)$ occuring in $\vec x$. 
  Therefore, $\trT < \trT$, which implies $i < i$ --- contradiction.
  \qed
\end{proof}

\begin{lemma}
  \label{lem:ON-disjoint-pre-trans}
  Let $\netN = (\Places, \Transitions, \Arcs, \markM[0])$ be an occurrence net.
  For all $\trT,\trTi \in \Transitions$, if $\trT \neq \trTi$ then
  $\pre{\trT} \cap \pre{\trTi} = \emptyset$.
\end{lemma}
\begin{proof}
  By contradiction, assume that $\placeP \in \pre{\trT} \cap \pre{\trTi}$
  with $\trT \neq \trTi$.
  Then, $\setenum{\trT,\trTi} \subseteq \post{\placeP}$, and hence
  $\card{\post{\placeP}} \geq 2$ ---
  contradiction with constraint~\ref{item:petri:onet:p}
  of the definition of occurrence nets.
  \qed
\end{proof}

\begin{lemma}
  \label{lem:occurrenceNet-LTS}
  Let $\markM$ be a reachable marking of an occurrence net $\netN$.
  Then:
  \begin{enumerate}
  \item     \label{lem:occurrenceNet-LTS:item-determinism}
    If $\markM \trans{\trT} \markMi$ and $\markM \trans{\trT} \markMii$,
    then $\markMi = \markMii$ (determinism).
  \item \label{lem:occurrenceNet-LTS:item-diamond}
    If $\markM \trans{\trT} \markMi$, $\markM \trans{\trTi} \markMii$
    and $\trT \neq \trTi$, then there exists $\markMiii$ such that 
    $\markMi \trans{\trTi} \markMiii$ and $\markMii \trans{\trT} \markMiii$ 
    (diamond property).
  \item \label{lem:occurrenceNet-LTS:item-distinct}
    If $\markM \trans{\trT}\;\trans{}^*\,\trans{\trTi}$ 
    then $\trT \neq \trTi$
    (linearity).
  \item \label{lem:occurrenceNet-LTS:item-acyclic}
    If $\markM \trans{\vec{\trT}}{\!\!}^* \;\markM$ then $\card{\vec{\trT}} = 0$
    (acyclicity).
  \end{enumerate}
\end{lemma}
\begin{proof}
  For \Cref{lem:occurrenceNet-LTS:item-determinism},
  by definition of the firing of transitions of Petri Nets it must be
  \(
  \markMi = \markM - \pre{\trT} + \post{\trT} = \markMii
  \).

  \medskip\noindent
  For \Cref{lem:occurrenceNet-LTS:item-diamond},
  since $\markM \trans{\trT} \markMi$ and $\markM \trans{\trTi} \markMii$,
  it must be:
  \begin{align*}
    & \pre{\trT} \subseteq \markM
    && \markMi = \markM - \pre{\trT} + \post{\trT}
    \\
    & \pre{\trTi} \subseteq \markM
    && \markMii = \markM - \pre{\trTi} + \post{\trTi}
  \end{align*}
  By \Cref{lem:ON-disjoint-pre-trans} it follows that
  $\trTi$ is enabled at $\markMi$, and $\trT$ is enabled at $\markMii$.
  Then, by definition of firing:
  \[
  \markMi \trans{\trTi} \markMi - \pre{\trTi} + \post{\trTi}\qquad \text{and} 
  \qquad \markMii \trans{\trT}\markMii - \pre{\trT} + \post{\trT}
  \]
  Then:
  \begin{align*}
    \markMi - \pre{\trTi} + \post{\trTi} 
    & = (\markM - \pre{\trT} + \post{\trT}) - \pre{\trTi} + \post{\trTi} \\
    & = (\markM - \pre{\trTi} + \post{\trTi}) - \pre{\trT} + \post{\trT} 
    && \text{(as $\pre{\trTi} \subseteq \markM$)} \\
    & = \markMii  - \pre{\trT} + \post{\trT}
  \end{align*}
  Hence, the thesis follows by choosing 
  $\markMiii = \markMi - \pre{\trTi} + \post{\trTi}$.

  \medskip\noindent
  \Cref{lem:occurrenceNet-LTS:item-distinct} follows directly
  by induction on the lenght of the reduction $\trans{}^*$,
  exploiting the fact that $\Arcs^*$ is a partial order.

  \medskip\noindent
  \Cref{lem:occurrenceNet-LTS:item-acyclic} follows by the fact that
  $\Arcs^*$ is a partial order.
  \qed
\end{proof}

\begin{lemma}
  \label{lem:occurrenceNet-reach}
  Let $\netN = (\Places, \Transitions, \Arcs, \markM[0])$ be an occurrence net,
  and let $\markM$ be a reachable marking, such that, 
  for some $\trT$, $\markMi$, $\markMii$:
  \begin{center}
    \begin{tikzcd}[column sep=normal, row sep=normal, negated/.style={
        decoration={markings,
          mark= at position 0.5 with {
            \node[transform shape] (tempnode) {$\slash$};
          }
        },
        postaction={decorate}
      }]
      \markM
      \ar[r]
      \ar[dr, start anchor=south east, end anchor=north west,"{\trT}", sloped, ""' near end]
      & {\!\!}^{n + 1} \;\; \markMi \ar[r, sloped, negated, "\trT"]
      & {} \\
      & \;\; \markMii
      & {}
    \end{tikzcd}
  \end{center}
  Then, $\markMii \trans{}^n \markMi$.
\end{lemma}
\begin{proof}
  By induction on $n$. 
  For the base case, it must be $\markM \trans{}^1 \markMi$,
  and hence $\markM \trans{\trTi} \markMi$ for some $\trTi$. Since $\markM 
  \trans{\trT} \markMii$, by the contrapositive of 
  \cref{lem:occurrenceNet-LTS:item-diamond} of \Cref{lem:occurrenceNet-LTS} 
  (diamond property) it follows $\trT = \trTi$, and so, by 
  \cref{lem:occurrenceNet-LTS:item-determinism} of \Cref{lem:occurrenceNet-LTS}
  (determinism) we have that $\markMi = \markMii$. Clearly: 
  \[
  \markMii \trans{}^0 \markMi
  \]
  For the inductive case, let $n = m + 1$, for some $m$.
  Then, for some $\trTi,\markMiii$: 
  \[
  \markM \trans{\trTi} \markMiii \trans{}^{m + 1} \markMi
  \] 
  If $\trT = \trTi$, then by 
  item~\ref{lem:occurrenceNet-LTS:item-determinism} 
  of~\Cref{lem:occurrenceNet-LTS} (determinism)
  it follows that $\markMiii = \markMii$,
  and so we have the thesis $\markMii \trans{}^{n} \markMi$.
  Otherwise, if $\trT \neq \trTi$,
  by item~\ref{lem:occurrenceNet-LTS:item-diamond}
  of~\Cref{lem:occurrenceNet-LTS} (diamond property), 
  there must exists
  $\markM[1]$ such that:
  \[
  \markMii \trans{\trTi} \markM[1] \;\;\text{and}\;\; 
  \markMiii \trans{\trT} \markM[1]
  \] 
  We are in the following situation:
  \begin{center}
    \begin{tikzcd}[column sep=normal, row sep=normal, negated/.style={
        decoration={markings,
          mark= at position 0.5 with {
            \node[transform shape] (tempnode) {$\slash$};
          }
        },
        postaction={decorate}
      }]
      \markMiii
      \ar[r]
      \ar[dr, start anchor=south east, end anchor=north west,"{\trT}", sloped, ""' near end]
      & {\!\!}^{m + 1} \;\; \markMi \ar[r, sloped, negated, "\trT"]
      & {} \\
      & \;\; \markM[1]
      & {}
    \end{tikzcd}
  \end{center}
  Since $m + 1 = n$, by the induction hypothesis:
  \[
  \markM[1] \trans{}^m \markMi
  \]
  Therefore, we have the thesis
  \[
  \markMii \trans{\trTi} \markM[1] \trans{}^m \markMi
  \qedhere
  \]
\end{proof}

\begin{lemma}
  \label{lem:PNet-leq-implies-not-independent}
  Let $(\txT,i),(\txTi,j)$ be transitions of
  $\PNet{\pswapWR}{\bcB}$, and let $\markM$ be a reachable marking. 
  Then:
  \[
  (\txT,i) < (\txTi,j)
  \;\;\text{and}\;\;
  \markM \trans{(\txT,i)}
  \quad\implies\quad
  \markM \nottrans{(\txTi,j)}
  \]
\end{lemma}
\begin{proof}
  By the construction in~\Cref{def:bc-to-pnet},
  since $(\txT,i) < (\txTi,j)$, then 
  $\placeP = ((\txT,i),(\txTi,j))$ is a place of the occurrence net,
  and 
  $\Arcs((\txT,i),\placeP) = 1$ and $\Arcs(\placeP,(\txTi,j)) = 1$.
  \qed
\end{proof}

\begin{definition}[\textbf{Independency}]
  \label{def:independent}
  Let $\netN$ be an occurrence net.
  We say that two transitions $\trT$ and $\trTi$ are \emph{independent}, 
  in symbols $\trT \independent \trTi$, 
  if $\trT \neq \trTi$ and there exists a reachable marking $\markM$ such that:
  \[
  \markM \trans{\trT} \;\;\text{and}\;\; \markM \trans{\trTi}
  \]
  We define the relation $\equivSeq$ 
  as the least congruence in the monoid $\Transitions^*$ such that,
  for all $\trT, \trTi \in \Transitions$:
  \(\;
  \trT \, I \, \trTi \implies \trT \trTi \equivSeq \trTi \trT
  \).
\end{definition}

\begin{lemma}
  \label{lem:PN-mstep-indepenent}
  Let $\netN$ be an occurrence net, with a reachable marking $\markM$. 
  If $\markM \trans{\trSU}$ then $\trT \independent \trTi$,
  for all $\trT \neq \trTi \in \trSU$.
\end{lemma}
\begin{proof}
  Since $\markM \trans{\trSU}$, then $\markM \trans{\trT}$ for all
  $\trT \in \trSU$.
  \qed
\end{proof}

\begin{lemma}
  \label{lem:PN-trace-equiv}
  Let $\netN$ be an occurrence net,
  and let $\markM$ be a reachable marking. 
  If $\markM \trans{\vec{\trT[1]}}{\!}^*\, \markMi$ and 
  $\markM \trans{\vec{\trT[2]}}{\!}^* \, \markMi$, 
  then $\vec{\trT[1]} \equivSeq \vec{\trT[2]}$.
\end{lemma}
\begin{proof}
  We proceed by induction on the length of the longest reduction among
  $\markM \trans{\vec{\trT[1]}}{\!}^* \markMi$ and 
  $\markM \trans{\vec{\trT[2]}}{\!}^* \markMi$.
  For the base case, the thesis is trivial as both $\vec{\trT[1]}$ and 
  $\vec{\trT[2]}$ are empty.
  For the inductive case, assume that $\vec{\trT[1]}$ 
  is longer or equal to $\vec{\trT[2]}$ (the other case is symmetric). 
  Let $\vec{\trT[1]} = \trT[1]\vec{\trTi[1]}$.
  We first show that $\vec{\trT[2]}$ is not empty. 
  By contradiction, if $\vec{\trT[2]}$ is empty, then $\markM = \markMi$. 
  But then, by \cref{lem:occurrenceNet-LTS:item-acyclic} of
  \Cref{lem:occurrenceNet-LTS} (acyclicity) it follows that $\vec{\trT[1]}$ is
  empty as well: contradiction. 
  Therefore, $\vec{\trT[2]} = \trT[2]\vec{\trTi[2]}$ for some $\trT[2]$ and 
  $\vec{\trTi[2]}$. 
  Clearly, $\vec{\trT[1]}$ is longer
  than $\vec{\trTi[1]}$ and $\vec{\trTi[2]}$. 
  Let $\markM \trans{\trT[1]} \markM[1]$ and $\markM \trans{\trT[2]} \markM[2]$.
  We have two subcases.
  \begin{itemize}
  \item If $\trT[1] = \trT[2]$, 
    by determinism (\Cref{lem:occurrenceNet-LTS}) 
    it follows that $\markM[1] = \markM[2]$.
    Let $\markMii = \markM[1]$.
    By the hypothesis of the~\namecref{lem:PN-trace-equiv}, we have
    $\markMii \trans{\vec{\trTi[1]}}{\!\!\!}^* \; \markMi$ and 
    $\markMii \trans{\vec{\trTi[2]}}{\!\!\!}^* \; \markMi$. 
    Then, by the induction hypothesis we have
    $\vec{\trTi[1]} \equivSeq \vec{\trTi[2]}$,
    and so the thesis $\trT[1] \vec{\trTi[1]} \equivSeq \trT[2] \vec{\trTi[2]}$
    follows since $\equivSeq$ is a congruence.
  \item If $\trT[1] \neq \trT[2]$, then by~\Cref{def:independent} it must be
    $\trT[1] \independent \trT[2]$. 
    By the diamond property (\Cref{lem:occurrenceNet-LTS}), 
    there exists $\markMii$
    such that $\markM[1] \trans{\trT[2]} \markMii$ and 
    $\markM[2] \trans{\trT[1]} \markMii$. 
    By linearity
    (item~\ref{lem:occurrenceNet-LTS:item-distinct} of \Cref{lem:occurrenceNet-LTS}), 
    $\markMi \nottrans{\;\;\trT[1]}$ and
    $\markMi \nottrans{\;\;\trT[2]}$. 
    By \Cref{lem:occurrenceNet-reach}, applied on $\markM[2]$, 
    there exists $\vec{\trT}$ such that
    $\markMii \trans{\vec{\trT}}{\!\!}^* \; \markMi$ 
    and $\card{\vec{\trT}} + 1 = \card{\vec{\trTi[2]}}$. 
    So, we are in the following situation:
    
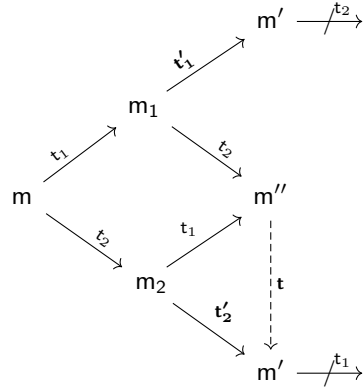
\begin{center}
      \begin{tikzcd}[column sep=normal, row sep=normal, negated/.style={
          decoration={markings,
            mark= at position 0.5 with {
              \node[transform shape] (tempnode) {$\slash$};
            }
          },
          postaction={decorate}
        }]
        & & \markMi 
        \ar[r, sloped, negated, "{\quad\trT[2]}"] & {} \\
        & \markM[1] 
        \ar[ru, sloped, "\vec{\trTi[1]}"] 
        \ar[rd, sloped, "{\trT[2]}"] & & \\
        \markM
        \ar[ru, sloped, "{\trT[1]}"]
        \ar[rd, sloped, "{\trT[2]}"]
        & 
        & \markMii
        \ar[dd, dashed, "\vec{\trT}"] 
        \\
        & \;\; \markM[2]
        \ar[ru, "{\trT[1]}"]
        \ar[rd, "\vec{\trTi[2]}"]
        \\
        & & \markMi 
        \ar[r, sloped, negated, "{\quad\trT[1]}"]
        & {}
      \end{tikzcd}
    \end{center}

    Therefore, we have that:
    \begin{align*}
      & \markM[1] \trans{\trT[2]} \trans{\vec{\trT}}{\!\!}^* \; \markMi
        \quad\text{and}\quad
        \markM[1] \trans{\vec{\trTi[1]}}{\!\!}^* \; \markMi
      \\
      & \markM[2] \trans{\trT[1]} \trans{\vec{\trT}}{\!\!}^* \; \markMi
        \quad\text{and}\quad
        \markM[2] \trans{\vec{\trTi[2]}}{\!\!}^* \; \markMi
    \end{align*}
    Note that  
    $\card{\trT[2]\vec{\trT}} = \card{\trT[1]\vec{\trT}} = \card{\vec{\trT}} + 1 
    = \card{\vec{\trTi[2]}} \leq  \card{\vec{\trTi[1]}} < \card{\vec{\trT[1]}}$.
    Hence, by applying the induction hypothesis twice:
    \[
    \trT[2]\vec{\trT} \equivSeq \vec{\trTi[1]} 
    \quad \text{and} \quad
    \trT[1]\vec{\trT} \equivSeq \vec{\trTi[2]}
    \]
    Then, since $\equivSeq$ is a congruence:
    \[
    \trT[1]\trT[2]\vec{\trT} \equivSeq \trT[1]\vec{\trTi[1]} \;\;\text{and}\;\;
    \trT[2]\trT[1]\vec{\trT} \equivSeq \trT[2]\vec{\trTi[2]}
    \]
    Since $\trT[1] \independent \trT[2]$, 
    then $\trT[1]\trT[2]\vec{\trT} \equivSeq \trT[2]\trT[1]\vec{\trT}$. 
    By transitivity of $\equivSeq$:
    \[
    \vec{\trT[1]} 
    \; = \;
    \trT[1]\vec{\trTi[1]} 
    \; \equivSeq \;
    \trT[1] \trT[2] \vec{\trT} 
    \; \equivSeq \;
    \trT[2] \trT[1] \vec{\trT} 
    \; \equivSeq \;
    \trT[2]\vec{\trTi[2]} 
    \; = \;
    \vec{\trT[2]}
    \tag*{\qed}
    \]
  \end{itemize}
\end{proof}

\begin{definition}
  \label{def:petri:trOfTrSU}
  For all sequences of transitions $\vec{\trT}$, we define 
  the set $\trOfTrSU{\vec{\trT}}$ of the transitions in $\vec{\trT}$ as:
  \[
  \trOfTrSU{\vec{\trT}} = 
  \setcomp{\trT}{\exists \vec{\trT[1]},\vec{\trT[2]}: 
    \vec{\trT} = \vec{\trT[1]}\trT\vec{\trT[2]}}
  \]
  and we extend $\trOfTrSU{}$ to step firing sequences $\vec{\trSU}$ as follows:
  \[
  \trOfTrSU{\vec{\trSU}} = 
  \bigcup \setcomp{\trSU}{\exists \vec{\trSU[1]},\vec{\trSU[2]}: 
    \vec{\trSU} = \vec{\trSU[1]}\trSU\vec{\trSU[2]}}
  \]
\end{definition}

\begin{lemma}
  \label{lem:equivSeq-implies-equal-sets}
  If $\vec{\trT} \equivSeq \vec{\trTi}$ then
  $\trOfTrSU{\vec{\trT}} = \trOfTrSU{\vec{\trTi}}$.
\end{lemma}
\begin{proof}
  Trivial by~\Cref{def:independent}.
  \qed
\end{proof}

\begin{lemma}
  \label{lem:equal-Tr-implies-equal-marking}
  Let $\netN$ be an occurrence net,
  and let $\markM$ be a reachable marking.
  If $\markM \trans{\vec{\trSU[1]}} \markM[1]$, 
  $\markM \trans{\vec{\trSU[2]}} \markM[2]$ and 
  $\trOfTrSU{\vec{\trSU[1]}} = \trOfTrSU{\vec{\trSU[2]}}$, then
  $\markM[1] = \markM[2]$.
\end{lemma}
\begin{proof}
  Since $\markM \trans{\vec{\trSU[1]}} \markM[1]$ and 
  $\markM \trans{\vec{\trSU[2]}} \markM[2]$, there exist
  sequentialisation $\vec{\trT[1]}$ of $\vec{\trSU[1]}$ 
  and $\vec{\trT[2]}$ of $\vec{\trSU[2]}$ such that
  $\markM \trans{\vec{\trT[1]}} \markM[1]$ and
  $\markM \trans{\vec{\trT[2]}} \markM[2]$.
  Since by hypothesis
  $\trOfTrSU{\vec{\trSU[1]}} = \trOfTrSU{\vec{\trSU[2]}}$, then
  $\trOfTrSU{\vec{\trT[1]}} = \trOfTrSU{\vec{\trT[2]}}$.
  We proceed by induction on the lenght of $\vec{\trT[1]}$.
  The base case is trivial, as $\markM = \markM[1] = \markM[2]$.
  For the inductive case, suppose $\vec{\trT[1]} = \trT[1]\vec{\trTi[1]}$,
  with $\card{\vec{\trTi[1]}} = n$. 
  By determinism, there exists a unique marking $\markMi[1]$
  such that $\markM \trans{\trT[1]} \markMi[1]$ (a single step).
  Since $\trOfTrSU{\vec{\trSU[1]}} = 
  \trOfTrSU{\vec{\trSU[2]}}$, it must be $\vec{\trT[2]} = \trT[2]\vec{\trTi[2]}$,
  with $\card{\vec{\trTi[2]}} = n$. 
  Let $\markMi[2]$ be the unique
  marking such that $\markM \trans{\trT[2]} \markMi[2]$ (a single step).

  \smallskip\noindent
  There are two subcases.
  \begin{itemize}
  \item If $\trT[1] = \trT[2]$, it must be $\markMi[1] = \markMi[2]$, and hence
    the thesis follows immediately by the induction hypothesis. 
  \item If $\trT[1] \neq \trT[2]$,
    by the diamond property 
    (\cref{lem:occurrenceNet-LTS:item-determinism} of~\Cref{lem:occurrenceNet-LTS}),
    there exists $\markMi$ such that $\markMi[1] \trans{\trT[2]} \markMi$ and
    $\markMi[2] \trans{\trT[1]} \markMi$. Since $\trT[2] \in
    \trOfTrSU{\vec{\trTi[1]}}$, by linearity
    (\cref{lem:occurrenceNet-LTS:item-distinct} of~\Cref{lem:occurrenceNet-LTS})
    it follows that
    $\markM[1] \nottrans{\;\;\trT[2]}$, and hence, 
    by applying \Cref{lem:occurrenceNet-reach} on $\markMi[1]$
    we obtain $\markMi \trans{\vec{\trTii[1]}} \markM[1]$ 
    for some $\vec{\trTii[1]}$. Summing up, we have that:
    \[\markMi[1] \trans{\vec{\trTi[1]}} \markM[1] \qquad \text{and} \qquad
    \markMi[1] \trans{\trT[2]\vec{\trTii[1]}} \markM[1]\]
    Then, by \Cref{lem:PN-trace-equiv}, 
    $\vec{\trTi[1]} \equivSeq \trT[2]\vec{\trTii[1]}$,
    and hence:
    \[
    \vec{\trT[1]} = \trT[1]\vec{\trTi[1]} \equivSeq 
    \trT[1]\trT[2]\vec{\trTii[1]}
    \]
    By \Cref{lem:equivSeq-implies-equal-sets}:
    \[
    \trOfTrSU{\vec{\trT[1]}} = \trOfTrSU{\trT[1]\trT[2]\vec{\trTii[1]}}
    \]
    Similarly, we can conclude that $\markMi \trans{\vec{\trTii[2]}} \markM[2]$ 
    for some $\vec{\trTii[2]}$ and that:
    \[
    \trOfTrSU{\vec{\trT[2]}} = \trOfTrSU{\trT[2]\trT[1]\vec{\trTii[2]}}
    \]
    Since $\trOfTrSU{\vec{\trT[1]}} = \trOfTrSU{\vec{\trT[2]}}$,
    we can conclude:
    \[\trOfTrSU{\vec{\trTii[1]}} = \trOfTrSU{\vec{\trTii[2]}}\]
    Since $\card{\vec{\trTii[1]}} = n - 1 < n + 1 = \card{\vec{\trT[1]}}$,
    the thesis follows by the induction hypothesis.
    \qed
  \end{itemize}
\end{proof}

\begin{lemma}
  \label{lem:PN-independent-implies-swap}
  Let $(\txT,i)$ and $(\txTi,j)$ be transitions of $\PNet{\pswapWR}{\bcB}$.
  Then:
  \[
  (\txT,i) \independent (\txTi,j)
  \quad \implies \quad
  \txT \pswapWR \txTi
  \]
\end{lemma}
\begin{proof}
  By~\Cref{def:independent},
  $(\txT,i) \independent (\txTi,j)$ implies that
  $(\txT,i) \neq (\txTi,j)$ and there exists some reachable marking
  $\markM$ such that $\markM \trans{(\txT,i)}$ and $\markM \trans{(\txTi,j)}$.
  By contradiction, assume that $\neg (\txT \pswapWR \txTi)$.
  Then, since $i < j$ or $j > i$, by \Cref{def:bc-to-pnet}
  we would have that $(\txT,i) < (\txTi,j)$ or $(\txTi,j) < (\txT,i)$.
  Then, by \Cref{lem:PNet-leq-implies-not-independent} 
  we obtain a contradiction.
  \qed
\end{proof}


\begin{definition}
  Let $\PNet{\relR}{\bcB} = (\Places, \Transitions, \Arcs, \markM[0])$.
  We define the function $\txOfTr{} : \Transitions \rightarrow \Tx$ as 
  \(
  \txOfTr{\txT,i} = \txT
  \).
  We then extend $\txOfTr{}$ to a function from steps to 
  multisets of transactions as follows:
  \[
  \txOfTr{\emptyset} = \emptymset 
  \qquad 
  \txOfTr{\trSU \cup \setenum{\trT}} = 
  \msetenum{\txOfTr{\trT}} + \txOfTr{\trSU}\;\;
  \]
  Finally, we extend $\txOfTr{}$ to finite sequences of steps as follows:
  \[
  \txOfTr{\bcEmpty} = \bcEmpty \qquad 
  \txOfTr{\trSU\vec{\trSU}} = \txOfTr{\trSU}\txOfTr{\vec{\trSU}}
  \]
\end{definition}

\begin{lemma}
  \label{lem:petri-equiv-par-block}
  Let $\PNet{\relR}{\bcB} = (\Places, \Transitions, \Arcs, \markM[0])$,
  and let $\vec{\trSU}$ be a step firing sequence. Then, for all $\WR{},\bcSt$:
  \[
  \msem{\WR{}}{\bcSt}{\vec{\trSU}} = \msem{\WR{}}{\bcSt}{\txOfTr{\vec{\trSU}}}
  \]
\end{lemma}

\begin{lemma}
  \label{lem:equivSeq-implies-equivStSeq}
  If $\vec{\trT} \equivSeq \vec{\trTi}$ in $\PNet{\pswapWR}{\bcB}$, then 
  $\txOfTr{\vec{\trT}} \equivStSeq[\pswapWR] \txOfTr{\vec{\trTi}}$.
\end{lemma}
\begin{proof}
  Define: 
  \[
  \equivSeq' = \setcomp{(\vec{\trT},\vec{\trTi})}{\txOfTr{\vec{\trT}} 
    \equivStSeq[\pswapWR] \txOfTr{\vec{\trTi}}}
  \] 
  It suffice to show that $\equivSeq\; \subseteq\; \equivSeq'$.
  Note that $\equivSeq'$ is a congruence satisfying:
  \[
  \txT \pswapWR \txTi \;\;\implies\;\; 
  (\txT,i)(\txTi,j) \equivSeq' (\txTi,j)(\txT,i)
  \]
  But then, by \Cref{lem:PN-independent-implies-swap}, 
  it follows that $\equivSeq'$
  also satisfies:
  \begin{equation}
    \label{lem:equivSeq-implies-equivStSeq:eq-1}
    (\txT,i) \independent (\txTi,j) \;\;\implies\;\; 
    (\txT,i)(\txTi,j) \equivSeq' (\txTi,j)(\txT,i)
  \end{equation}
  Since $\equivSeq$ is the smallest congruence satisfying
  \cref{lem:equivSeq-implies-equivStSeq:eq-1}, we have that
  $\equivSeq\; \subseteq\; \equivSeq'$, as required.
  \qed
\end{proof}

\thbctopnet*
\begin{proof}
  For item~\ref{th:bc-to-pnet:confluence},
  assume that $\markM[0] \trans{\vec{\trSU}} \markM$ and 
  $\markM[0] \trans{\vec{\trSUi}} \markM$. 
  A standard result from Petri nets theory ensures that
  there exists sequentializations $\vec{\trT}$ of $\vec{\trSU}$
  and $\vec{\trTi}$ of $\vec{\trSUi}$ such that:
  \[
  \markM[0] \trans{\vec{\trT}} \markM
  \quad \text{and} \quad
  \markM[0] \trans{\vec{\trTi}} \markM
  \]
  By \Cref{lem:PN-trace-equiv},
  it must be $\vec{\trT} \equivSeq \vec{\trTi}$.
  Then, by \Cref{lem:equivSeq-implies-equivStSeq}: 
  \[\txOfTr{\vec{\trT}} \equivStSeq[\pswap{}{}] \txOfTr{\vec{\trTi}}\]
  By \Cref{lem:equivRel-implies-equiv}:
  \[
  \semBc{\txOfTr{\vec{\trT}}}{\bcSt} = \semBc{\txOfTr{\vec{\trTi}}}{\bcSt}
  \]
  By \Cref{lem:PN-mstep-indepenent,lem:PN-independent-implies-swap},
  it follows that all multisets of transactions in $\txOfTr{\vec{\trSU}}$,
  as well as those in $\txOfTr{\vec{\trSUi}}$,
  are strongly swappable.
  %
  %
  Therefore, by \Cref{th:seq-union}:
  \[
  \msem{\WRmin{}}{\bcSt}{\txOfTr{\vec{\trSU}}} = 
  \msem{\WRmin{}}{\bcSt}{\txOfTr{\vec{\trSUi}}}
  \]
  Then, by \Cref{lem:petri-equiv-par-block}:
  \[
  \msem{\WRmin{}}{\bcSt}{\vec{\trSU}} = 
  \msem{\WRmin{}}{\bcSt}{\vec{\trSUi}}
  \]
  \medskip
  For item~\ref{th:bc-to-pnet:maximal},
  let $\vec{\trSUi} = \msetenum{\trT[1]}\hdots\msetenum{\trT[n]}$,
  where $n = \card{\bcB}$. It is trivial to see that $\vec{\trSUi}$ is maximal 
  and that $\bcB \seqn \txOfTr{\vec{\trSUi}}$. 
  By \Cref{th:seq-union}: 
  \[\msem{\WRmin{}}{\bcSt}{\txOfTr{\vec{\trSUi}}} = \semBc{\bcB}{\bcSt}\]
  Since $\vec{\trSU}$ and $\vec{\trSUi}$ are both maximal, by 
  \Cref{lem:equal-Tr-implies-equal-marking} and 
  by item~\ref{th:bc-to-pnet:confluence}, it follows that:  
  \[
  \msem{\WRmin{}}{\bcSt}{\vec{\trSUi}} = \msem{\WRmin{}}{\bcSt}{\vec{\trSU}}
  \]
  Since $\msem{\WRmin{}}{\bcSt}{\vec{\trSUi}} = 
  \msem{\WRmin{}}{\bcSt}{\txOfTr{\vec{\trSUi}}}$ 
  (by \Cref{lem:petri-equiv-par-block}), and so:
  \[
  \msem{\WRmin{}}{\bcSt}{\vec{\trSU}} = \semBc{\bcB}{\bcSt}
  \tag*{$\qed$}
  \]
\end{proof}

  \newpage
  \section{Full implementation of the ERC721 Token}
\label{sec:full-erc721}

\begin{lstlisting}[language=solidity,numbersep=10pt]
pragma solidity >= 0.4.2;

contract Token {
    
    mapping(uint256 => address) owner;
    mapping(uint256 => address) approved;
    mapping(uint256 => bool) exists;
    mapping(address => uint256) balance;
    mapping (address => mapping (address => bool)) operatorApprovals;
    
    function ownerOf(uint256 tokenId) external view returns (address) {
    	require (exists[tokenId]); 
    	require (owner[tokenId] != address(0));
    	return owner[tokenId];
    }
    
    function balanceOf(address addr) external view returns (uint256) {
        require (addr != address(0));
        return balance[addr];
    }

    function approve(address addr, uint256 tokenId) external {
	    require (exists[tokenId]);
	    require (owner[tokenId] == msg.sender && addr != msg.sender);
	    approved[tokenId] = addr;
    }

    function setApprovalForAll(address operator, bool isApproved) external {
        operatorApprovals[msg.sender][operator] = isApproved;
    }
	
    function getApproved(uint256 tokenId) external view returns (address) {
        require (exists[tokenId]);
        return approved[tokenId];
    }

    function isApprovedForAll(address addr, address operator) external view returns (bool) {
        return operatorApprovals[addr][operator];
    }
	
    function transferFrom(address from, address to, uint256 tokenId) external {
        require (exists[tokenId]);
        require (from == owner[tokenId] && from != to);
        require (to != address(0));
        
        if (from == msg.sender || operatorApprovals[from][msg.sender] || approved[tokenId] == msg.sender) {
            owner[tokenId] = to;
            approved[tokenId] = address(0);
            balance[from] -= 1;
            balance[to] += 1;
        }
    }
    
    function mint(address to, uint256 tokenId) external {
        require (!exists[tokenId]);
        require (to != address(0));
        
        exists[tokenId] = true;
        owner[tokenId] = to;
        balance[to] += 1;
    }
}
\end{lstlisting}

We have the following safe approximations of read/written keys of the \code{Token} contract
(assuming the straightforward generalization of the semantic domain of~\Cref{sec:transactions} 
to deal with multiple key-value stores:
\[
\begin{array}{rcll}
  \QmvA[r]^1
  & = & \setenum{\codeVal{exists[1]}, \codeVal{owner[1]}, \codeVal{balance[\pmvA]}, \codeVal{balance[\pmv{P}]}}
  & \models^{r} \txT[1]
  \\
  \QmvA[w]^1
  & = & \setenum{\codeVal{owner[1]}, \codeVal{balance[\pmvA]}, \codeVal{balance[\pmv{P}]}} 
  & \models^{w} \txT[1]
  \\
  \QmvA[r]^2
  & = & \emptyset
  & \models^{r} \txT[2]
  \\
  \QmvA[w]^2
  & = & \setenum{\codeVal{operatorApprovals[A][B]}} 
  & \models^{w} \txT[2]
  \\
  \QmvA[r]^3
  & = & \setenum{\codeVal{exists[2]}, \codeVal{owner[2]}, \codeVal{balance[\pmvA]}, \codeVal{balance[\pmv{Q}]}}
  & \models^{r} \txT[3]
  \\
  \QmvA[w]^3
  & = & \setenum{\codeVal{owner[2]}, \codeVal{balance[\pmvA]}, \codeVal{balance[\pmv{Q}]}}
  & \models^{w} \txT[3] 
  \\
  \QmvA[r]^4
  & = & \setenum{\codeVal{exists[1]}, \codeVal{owner[1]}, \codeVal{balance[\pmv{P}]}, \codeVal{balance[\pmvB]}}
  & \models^{r} \txT[4]
  \\
  \QmvA[w]^4
  & = & \setenum{\codeVal{owner[1]}, \codeVal{balance[\pmv{P}]}, \codeVal{balance[\pmv{B}]}} 
  & \models^{w} \txT[4]
  \\
\end{array}
\]

}
{}

\end{document}
